\documentclass[twocolumn,superscriptaddress, floatfix,amsmath,amssymb,nofootinbib,prx]{revtex4-2}
\usepackage[english]{babel}
\usepackage{graphicx}
\usepackage{dcolumn}
\usepackage{bm}
\usepackage{verbatim}
\usepackage{mathrsfs}
\usepackage{natbib}
\usepackage{hyperref}
\usepackage{etoolbox}
\usepackage{orcidlink} 
\usepackage{amsthm}
\usepackage{bbm}
\usepackage[ruled,vlined]{algorithm2e}
\usepackage{booktabs}
\usepackage{multirow}
\usepackage{tabularx}
\usepackage{supertabular}
\usepackage[dvipsnames]{xcolor}
\usepackage[utf8]{inputenc}
\pdfoutput=1

\newcommand{\ket}[1]{\left| {#1} \right\rangle}
\newcommand{\bra}[1]{\left\langle {#1} \right|}
\newcommand{\ematrix}[3]{\left\langle {#1} \left|{#2}\right|{#3}\right\rangle}
\newcommand{\braket}[2]{\left\langle {#1}\left|{#2}\right.\right\rangle}
\newcommand{\proj}[2]{\left| {#1} \right\rangle\!\left\langle {#2} \right|}

\newcommand{\ii}{\mathrm{i}}

\newcommand{\tr}{\operatorname{Tr}}

\newcommand{\bfrac}[2]{\left(\frac{#1}{#2}\right)}

\newcommand{\dimHacc}{\rm{dim}\mathcal{H}}
\newcommand{\Hacc}{\mathcal{H}}
\newcommand{\setClass}{\Gamma}
\newcommand{\setClasses}{\ensuremath{\mathbf{\Gamma}}}
\newcommand{\setInputs}{\ensuremath{\mathcal{X}}}
\newcommand{\setOutputs}{\ensuremath{\mathcal{Y}}}
\newcommand{\intOutputs}{K}
\newcommand{\setZp}{\ensuremath{\mathbb{Z}_p}}
\newcommand{\setSetMods}{\ensuremath{\mathbf{\Xi}}}
\newcommand{\setMods}{\Xi}
\newcommand{\strInput}{\vec{w}}
\newcommand{\charInput}{x}
\newcommand{\charOutput}{y}
\newcommand{\fnClass}[1]{C({#1})}
\newcommand{\fnGuess}[1]{\hat{C}({#1})}

\newcommand{\vecOutput}{Y}
\newcommand{\intLen}{T}
\newcommand{\charEnd}{\Omega}
\newcommand{\agent}{\ensuremath{\mathcal{A}}}
\newcommand{\taskTSC}{\ensuremath{\textrm{TSC}}}
\newcommand{\taskMSC}{\textrm{BPG}($p,\setSetMods$)}

\newcommand{\taskMCC}{\textrm{MCC}($p,\setSetMods$)}

\newcommand{\setData}{\ensuremath{\mathcal{D}}}
\newcommand{\fnLoss}[1]{\ensuremath{L\left({#1}\right)}}
\newcommand{\fnELoss}[1]{\ensuremath{\mathcal{L}\left({#1}\right)}}
\newcommand{\setStatesBasis}{S_b}
\newcommand{\setpStates}{S}
\newcommand{\eRisk}{\mathcal{L}}
\newcommand{\stream}{\vec{\bf{w}}}

\newcommand{\Nmodx}[1]{N \,\,{\rm mod }\,\,{#1}}

\newcommand{\causalstate}{s}
\newcommand{\fnTransition}{\mathcal{T}}
\newcommand{\fnTransOne}{G}
\newcommand{\fnMeasure}{\mathcal{M}}
\newcommand{\tupleAgent}{(\setInputs,\setOutputs,\setpStates,\causalstate_0,\fnTransition)}

\newcommand{\setDiffs}{\Delta}
\newcommand{\tupleAgentShort}{\ensuremath{(\setpStates,\causalstate_0,(\fnTransOne,\fnMeasure))}}
\newcommand{\intPeriod}{h}

\definecolor{blue}{rgb}{0,0.2,1}

\definecolor{red}{rgb}{0.9,0,0}

\definecolor{green}{rgb}{0,1,0}

\theoremstyle{plain}
\newtheorem{definition}{Definition}
\newtheorem{theorem}{Result}
\newtheorem{problem}{Problem}

\newtheorem{lemma}{Lemma}[theorem]

\theoremstyle{definition}

\newcommand{\streamInput}{\mathbf{w}}

\newcommand{\be}{\begin{equation}}
\newcommand{\ee}{\end{equation}}
\newcommand{\SMclassical}{\ref{supp:sec:classical}}
\newcommand{\SMconstruct}{\ref{sec:construction}}
\newcommand{\SMunitary}{\ref{supp:sec:unitary}}
\newcommand{\SMoptimal}{\ref{supp:sec:optimal}}
\newcommand{\SMdist}{\ref{sec:naturalguess}}

\begin{document}
\title{Single-shot online sequence classification with unbounded quantum memory advantage   }
\author{Keith K. Ng\,\orcidlink{0000-0002-6552-8549}}
\affiliation{Nanyang Quantum Hub, School of Physical and Mathematical Sciences, Nanyang Technological University, Singapore}
\affiliation{Data Cybernetics, Vordere Mühlgasse 189, Landsberg am Lech, 86899, Germany}

\author{Haochen Jay Li\,\orcidlink{0009-0002-7025-1879}}
\affiliation{Nanyang Quantum Hub, School of Physical and Mathematical Sciences, Nanyang Technological University, Singapore}
\affiliation{Centre for Quantum Technologies, Nanyang Technological University, Singapore}
\author{Mile Gu\,\orcidlink{0000-0002-5459-4313}}
\email{mgu@quantumcomplexity.org}
\affiliation{Nanyang Quantum Hub, School of Physical and Mathematical Sciences, Nanyang Technological University, Singapore}
\affiliation{Centre for Quantum Technologies, Nanyang Technological University, Singapore}
\author{Jayne Thompson\,\orcidlink{0000-0002-3746-244X}}
\email{jayne.thompson@ntu.edu.sg}
\affiliation{Centre for Quantum Technologies, Nanyang Technological University,  Singapore}
\affiliation{College of Computing and Data Sciences, Nanyang Technological University,  Singapore}
\affiliation{School of Physical and Mathematical Sciences, Nanyang Technological University,  Singapore}

\begin{abstract}
An agent monitors a complex environment, receiving one observation at each time step and eventually deciding how to label the resulting sequence. Does the sequence indicate an anomaly, and if so, of what type? Does it signal market instability, and to what degree? This is the setting of online multi-class classification: the input arrives sequentially, the full history is never available at once, and the agent must retain any past information relevant to the eventual decision. As the environment becomes more complex, the memory needed to track this information can grow rapidly without bound. Here, we introduce families of such multi-class classification games and show that any exact classical agent requires memory that grows without bound, whereas exact quantum agents can solve all tasks in the family with bounded memory. This separation is sharp: any classical agent using less than the required memory, under suitable input distributions, performs arbitrarily close to random guessing. Moreover, our quantum constructions are provably memory minimal, allowing us to derive the exact classical and quantum memory complexities to perform such tasks. In doing so, we establish an unbounded separation between classical and quantum memory cost for online multi-class classification.
\end{abstract}

\maketitle

In online classification, an agent receives an input at each time step. Over time, these inputs form a string that the agent must classify into one of $K$ disjoint classes. Since the input arrives sequentially and the full history is never available all at once, at each time step the agent must determine which past features remain relevant for later decisions. Such settings arise whenever a device must classify streaming data on the fly, for example when monitoring sensor streams for anomalies, reacting to evolving financial data, or detecting and classifying errors from sequential measurements~\cite{angelatos2021reservoir}. In this context, it is natural to ask what and how much of the data observed so far an agent needs to retain to classify correctly. In addition to the obvious benefit of reduced memory resources during model inference, the `how much'' aligns with complexity-theoretic measures designed to identify what input-output behavior is most complex~\cite{crutchfield1989inferring,barnett2015computational},  while the former can guide more effective learning and generalization in data-scarce settings~\cite{lv2022causality,jonschkowski2015learning,banchi2021generalization}.

\emph{Can quantum mechanics fundamentally reduce this memory cost}? At first, this appears uncertain. Quantum models can achieve striking expressive power by encoding data into continuously varying quantum states~\cite{li2023quantum,grant2018hierarchical,perez2020data,dutta2022,sornsaeng2024quantum,napat2021}. However, classification is typically decided from quantum measurement expectation values - requiring many repeated runs during inference to estimate the correct label. The memory cost must account for the number of copies needed to support that estimation, so expressive power need not imply a memory advantage. By contrast, in the single-shot setting studied here, the agent must make its decision based on a single measurement~\cite{recio2025single}. If a binary feature can lead to different future classifications, optimal performance demands that an agent store different values of this feature in fully distinguishable states. This remains true even if our agents are quantum, and thus there is no obvious reason to anticipate a quantum memory advantage for lossless classification.

Here, we demonstrate otherwise: \emph{quantum agents can exhibit unbounded memory advantage in online, single-shot multi-class classification}. This is true for any number of classes. For each $K \geq 2$, we exhibit families of online classification tasks involving sequential data generated by a classical system with $p$ possible hidden configurations. To achieve perfect classification, a classical agent must maintain a perfect mental model of the system whenever $p$ is prime. The classical memory complexity of the task thus scales without bound with $p$. Moreover, any classical agent operating below the classical threshold can be forced, under suitable input distributions, to perform arbitrarily close to random guess. In contrast, we present quantum agents that achieve perfect classification using memory that is bounded above by a finite value that does not scale with $p$. Our quantum agents are provably memory-minimal and thus enable exact determination of the task's quantum memory complexity. Thus, we establish a clear unbounded separation between classical and quantum memory complexity in online data-stream classification during model inference.

\begin{figure*}[tbp]
    \centering
    \includegraphics[width=0.95\textwidth]{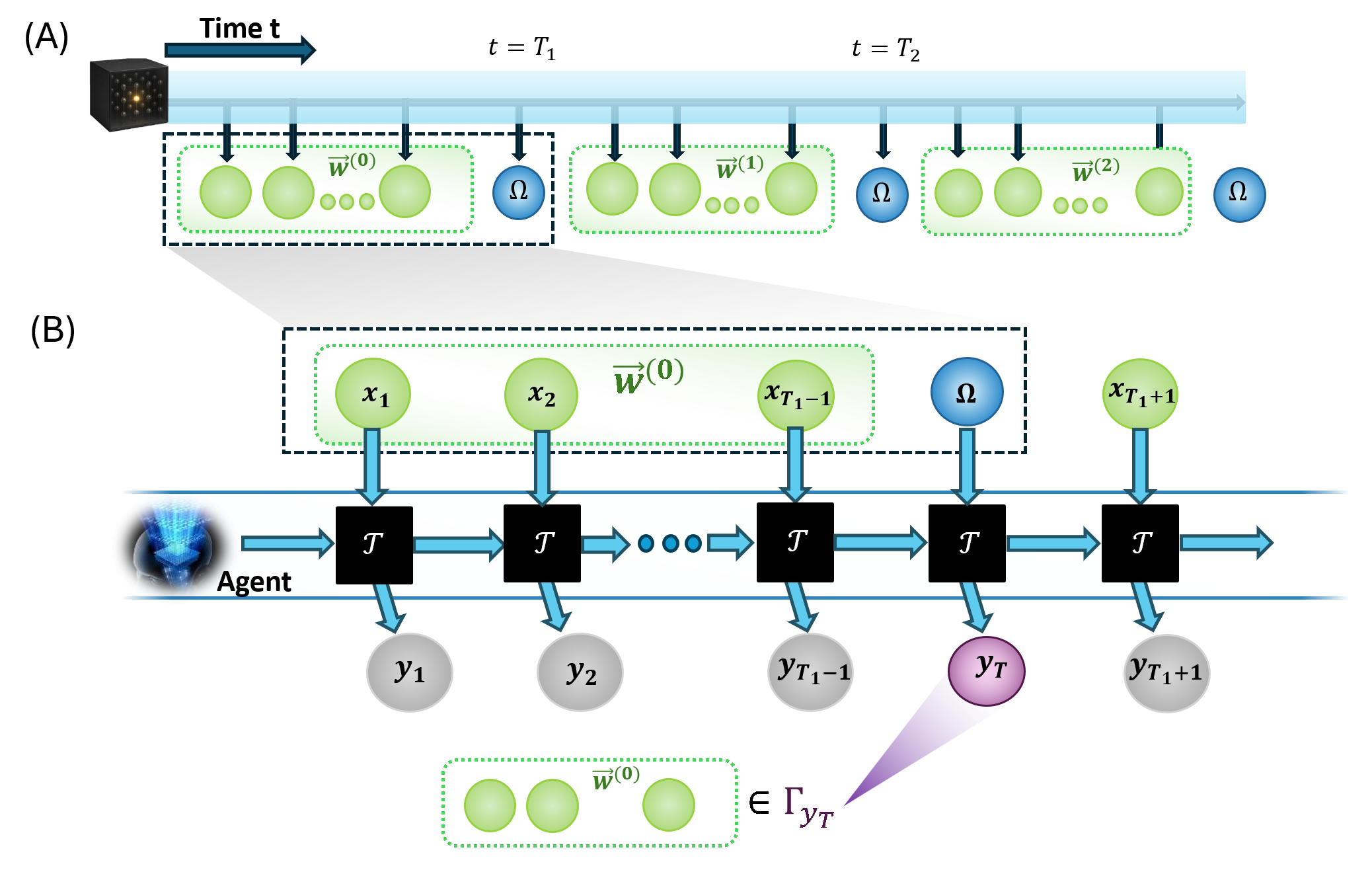}
  \caption{\textbf{Online Sequence Classification.} (A) An environmental system emits an input symbol $\charInput_t$ at each time $t$. Consecutive input symbols (green orbs) form words, which are separated by the delimiter symbol $\charEnd$ (blue orbs). (B) An online sequence-classification agent receives each $\charInput_t$ as it is emitted. To classify successfully, whenever the agent receives a delimiter $\charInput_{T_k}=\charEnd$, its output $\charOutput_{T_k}$ (purple orb) must be the label of the class containing the preceding word $\strInput^{(k-1)}=\charInput_{T_{k-1}+1}\cdots\charInput_{T_k-1}$. To do this, the agent generally requires a memory system $S$, which interacts sequentially with each input $\charInput_t$ through the corresponding operation $\mathcal{T}^{\charInput_t}$. At all other times, when $\charInput_t\neq\charEnd$, the agent may output any $\charOutput_t$ (gray orbs) without penalty.}
    \label{fig:transducercircuit}
\end{figure*}

\section{Framework}

\textbf{Online Classification}. We consider an environmental system that emits a sequence of symbols $\charInput_t \in \setInputs$ at discrete time steps $t \in \mathbb{N}$, starting from $t=0$, where the finite input alphabet $\setInputs$ contains a special string-delimitation symbol $\charEnd \in \setInputs$. Whenever $x_t=\charEnd$, the current input word terminates. At the next time step a new word begins. For $k \geq 1$, let $T_k$ denote the time of the $k^{\mathrm{th}}$ occurrence of $\charEnd$, and set $T_0=-1$. The resulting data stream $\streamInput=\strInput^{(0)}\charEnd\,\strInput^{(1)}\charEnd\,\strInput^{(2)}\charEnd\cdots$ can then be decomposed into a sequence of variable-length words $\strInput^{(k)} = \charInput_{T_k+1}\charInput_{T_k+2}\cdots \charInput_{T_{k+1}-1}$ punctuated by $\charEnd$'s.

Each specific classification task is defined by a collection of $K$ disjoint classes $\setClasses=\{\setClass_j\}_{j=1}^{\intOutputs}$, where $\setClass_j$ is the set of variable-length words associated with the class label $j$. Each word $\strInput$ belongs to at most one class, though not all words necessarily belong to a class. We then challenge an agent to determine, upon receiving each delimiter, which class the preceding word belongs to. That is, upon receiving the $k^{\mathrm{th}}$ delimiter at time $T_k$, the agent must guess the label of the preceding word $\strInput^{(k-1)}$. \emph{Optimal performance} occurs when the agent is always successful, i.e., it outputs $j$ whenever the preceding word $\strInput^{(k-1)} \in \setClass_j$.

Non-optimal play occurs when an agent can make incorrect guesses with non-zero probability. Let $P_{\mathrm{fail}}(\strInput)$ be the probability that the agent outputs some $j' \neq j$ when $\strInput$ lies in some class $\setClass_j$ and $0$ otherwise. We quantify performance by the loss $\fnELoss{\agent} = \mathbb{E}_{\strInput \sim \setData}[P_{\mathrm{fail}}(\strInput)]$, where the expectation is taken over a given distribution $\setData(\strInput)$ on input words. Clearly, an optimally performing agent has zero loss.

\textbf{Agents and Memory}. Here, we require classification to be performed \emph{online}. That is, at each time step $t$, the agent receives the current input symbol $\charInput_t \in \setInputs$, updates its internal state, and upon receiving a delimiter symbol $\charEnd$, must output a guess for the class of the preceding word in real time. Because the agent cannot revisit the past, it must instead store in memory any features of the observed history that remain relevant for eventual classification. This need for memory already appears in simple examples: if words $001$ and $101$ belong to different classes, the agent must retain the value of the first input bit.

More formally, we describe an online agent $\agent$ by a physical memory system $S$ initialized in some default state $s_{0} = s_{\mathrm{init}}$ at $t=0$. At each time step $t$, the agent's dynamics is governed by a policy: a collection of input-dependent physical operations $\mathcal{T} = \{\mathcal{T}^x\}_{x \in \mathcal{X}}$ on $\mathcal{S}$, that (i) emits some output $y \in \mathcal{Y}$ and (ii) updates the memory from state $s_t$ to $s_{t+1}$. Each agent $\mathcal{A}$ can then be formally defined by the tuple $(S, s_0, \mathcal{T}, \mathcal{X},\mathcal{Y})$.

In line with other single-shot settings, we quantify the size of a memory $S$ by its memory dimension $M$, defined as the maximum number of perfectly distinguishable memory states that can be stored~\cite{vieira2022temporal,elliott2020extreme,konig2005power,barnett2015computational}. For a quantum agent $\agent_Q$, this is the dimension of its memory Hilbert space $\mathcal{S}$; such that its memory state is some arbitrary density operator on $\mathcal{S}$. For a classical agent $\agent_C$, the memory is a finite-state register with $M$ distinct internal configurations, so that its state is a probability distribution over these configurations. Thus, in both the classical and quantum settings, $M$ measures the number of perfectly distinguishable memory states that can be stored. In particular, an $n$-bit memory and an $n$-qubit memory both have memory dimension $2^n$.

\section{Results}

Given a classification game specified by a collection of disjoint nonempty classes $\setClasses=\{\setClass_j\}_{j=1}^{\intOutputs}$ with $K \geq 2$, we define its classical and quantum memory complexities, denoted respectively by $M_C$ and $M_Q$, as the minimal memory cost required to perform the task with zero loss using classical and quantum agents. The central result of this manuscript can then be formally stated.

\begin{theorem}
\label{thm:main}
There exists an unbounded separation between the classical and quantum memory complexities $M_C$ and $M_Q$. In particular, there exists a family of classification games parametrized by a prime number $p$ for which both $M_C$ and $M_Q$ can be evaluated exactly. We show that $M_C$ grows without bound as $p \rightarrow \infty$, while $M_Q$ remains bounded for all $p$. Moreover, for every $\epsilon>0$, any classical agent with memory
size $M<M_C$ can be subjected to an input distribution under which it performs no more than $\epsilon$ better than random guessing.
\end{theorem}

We justify this result by first introducing a collection of \emph{Prize-wheel Games} encompassing families of classification games. Each family is a variant of an agent tracking a rotating wheel with $p$ possible configurations. We then proceed to (1) derive the ultimate performance limits of classical agents in this game, showing that perfect play requires memory $M_C = p$, while any smaller memory reduces performance arbitrarily close to that of random guessing among the $K$ possible labels under suitable input distributions; (2) construct quantum agents for such games whose memory remains bounded for all $p$, thereby establishing an upper bound for $M_Q$; and (3) prove that this construction is optimal, thereby determining $M_Q$ exactly.

\begin{figure}[tbp]
    \centering
    \includegraphics[width=0.38\textwidth]{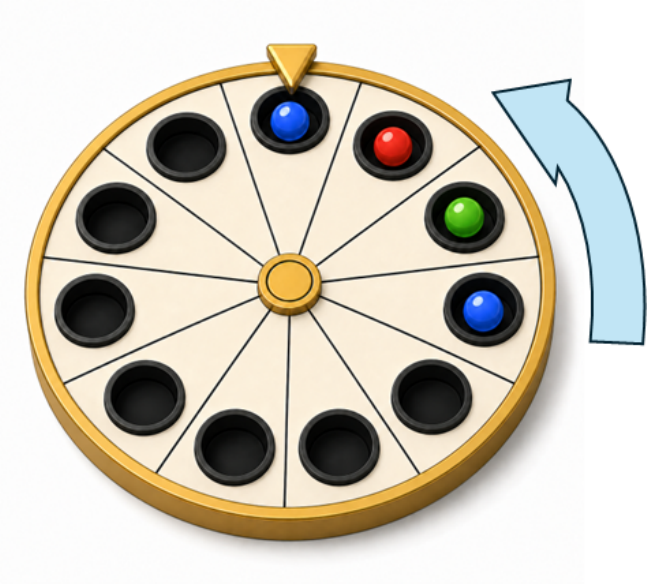}
    \caption{\textbf{Blackbox Prize-wheel Games.} Consider a prize-wheel with \(p\) visually indistinguishable partitions. Each partition contains a cup that may hold a colored marble. The initial configuration of the marbles is publicly announced as part of the game rules. At each round, the wheel is spun. Our agent can only observe the notches it has advanced.
    Upon receiving a stopping signal, the agent must name a color and is rewarded if the partition aligned with the designated marker contains a marble of that color. In the game pictured, classical agents require a memory dimension of $11$ to guarantee performance substantially better than random guessing. In contrast, a quantum agent needs a memory dimension of 5 for perfect play.}
    \label{fig:ratchet}
\end{figure}

\textbf{The Blackbox Prize-wheel Game}. Our collection of games can be understood through the spinning prize wheel of Fig.~\ref{fig:ratchet}. In full generality, the game consists of a wheel with $p$ visually identical partitions, for some prime number $p$, each of which may contain at most one marble. The marbles can take any of $K \ge 2$ possible colors, described by the color set $\setOutputs$. The wheel is fixed to an axis and can rotate relative to a prize marker. We use the phrase `rotate the wheel by $x$' as shorthand for rotating the prize wheel counterclockwise by $x$ partitions, i.e., by an angle of $2x\pi/p$.

We label the partitions by integers $0 \leq j \leq p-1$ such that partition $0$ is the one initially flagged by the prize marker, and partition $j$ is the $j^{th}$ partition from partition $0$, counting clockwise. We can then specify the initial configuration of the wheel by an ordered list $\mathbf{c}_{\mathrm{def}}=(c_0,c_1,\ldots,c_{p-1})$, where $c_j = y$ for some $y\in\setOutputs$ if partition $j$ contains a marble of color $y$ and $c_j= \emptyset$ if partition $j$ is empty. We refer to this as the \emph{default configuration}.
Equivalently, for each color $y\in\setOutputs$, we may define $\setMods_y \subseteq \mathbb{Z}_p$ as the set of partition numbers containing a marble of color $y$, i.e. $\setMods_y=\{j\in\mathbb{Z}_p : c_j=y\}$. The default configuration is then fully specified by the collection $\setSetMods = \{\setMods_y\}$.

At each time step, the game-master Victor either rotates the wheel by some $x_t \in \mathbb{Z}_p$, or issues a challenge for the agent to guess the color of the marble currently under the prize marker. The agent incurs no loss if it guesses correctly when the marked partition is non-empty, or if it makes any guess when the marked partition is empty. After each challenge, the wheel resets to the default configuration, and the procedure repeats. Thus, we can think of $p$ as the complexity of our environment, designating the minimal memory needed to simulate the world model our agent lives in~\cite{johnson1983mental}.

This collection of games fits exactly into our online classification framework. Each rotation corresponds to an input symbol $x_t \in \mathbb{Z}_p$, while each challenge corresponds to the delimiter input $\charEnd$. Thus, the agent receives a stream $\streamInput=\strInput^{(0)}\charEnd\,\strInput^{(1)}\charEnd\,\strInput^{(2)}\charEnd\cdots$, where each word $\strInput$ is a finite sequence of rotations between two successive challenges. The color of the marble is then entirely defined by the total rotation of the wheel relative to the default configuration - in other words, by properties of each word $\strInput^{(k)} = x_{T_k+1}x_{T_k+2}\dots x_{T_{k+1}-1}$. Setting $\setClass_y$ to be the class of words that lead to a marble of color $y$, we find
\begin{equation}\label{eq:MSClabel}
\strInput^{(k)}  \in \setClass_y
\qquad\textrm{iff}\qquad
\left(\sum_{i= T_k + 1 }^{T_{k+1} -1} \charInput_i\right)\!\!\!\!\pmod p \in \setMods_y .
\end{equation}
In other words, after a word $\strInput^{(k)} = x_{T_k+1}x_{T_k+2}\dots x_{T_{k+1}-1}$ of wheel rotations, the color under the prize marker depends only on the cumulative rotation modulo $p$. Each default configuration therefore defines a different classification task, which we abbreviate as $\taskMSC$.

\textbf{Classical Limits}. We begin by establishing the minimal memory needed to play $\taskMSC$ with zero loss.

\begin{theorem}
\label{thm:classical-memory}
The classical memory complexity $M_C$ of $\taskMSC$ is exactly $p$, provided that the configuration $\setSetMods$ has at least two classes.
\end{theorem}

We sketch the intuition here and leave the details to Appendix~\SMclassical. First, observe that the prize wheel has at most $p$ possible rotated configurations we can reach during any particular game. Tracking these configurations in mutually distinguishable states is sufficient for perfect classification. Thus we can set $\agent_C$ to have memory states $\mathcal{S}_C = \{s_j\}_{j=0}^{p-1}$ in one-to-one correspondence with these configurations, where $s_j$ labels the configuration obtained by rotating $\mathbf{c}_{\mathrm{def}}$ by $j$. The agent then solves this problem by (i) initializing in state $s_0$, (ii) updating $s_j$ to $s_{(j+x)\bmod p}$ on input $x$, and
(iii) upon input $\Omega$, outputting the color associated with
the current configuration $s_j$ and resetting to $s_0$. As such, $M_C \leq p$.

To see that no smaller machine can succeed, note that if $M_C<p$, then the agent's memory cannot distinguish all $p$ configurations of the prize wheel with certainty. There must therefore be two configurations, say, $s_j$ and $s_{j+\delta}$ (with all additions understood modulo $p$) the agent cannot distinguish perfectly. If the wheel is in one of these configurations and the agent subsequently receives input $1$, it can no longer determine whether the wheel has moved to $s_{j+1}$ or to $s_{j+1+\delta}$. By the data-processing inequality~\cite{beaudry2011intuitive}, this ambiguity cannot decrease under further inputs. Iterating this observation shows that the agent gradually loses track of the wheel's configuration altogether. In fact, this loss of information allows us to constrain the performance of any such memory-limited classical agent (see Appendix~\SMclassical\ and \SMdist).

\begin{theorem}
\label{thm:classical-loss}
Suppose a classical agent $\agent_C$ has memory cost $M<p$. Then, for any $\epsilon>0$, there exists a distribution $\setData_{\epsilon}(\strInput)$ such that the expected loss $\fnELoss{\agent_C}$ of the agent on $\taskMSC$ satisfies
\begin{equation}
 \fnELoss{\agent_C} \ge  1 - \frac{1}{K}-\epsilon,
\end{equation}
where $K$ is the number of classes in $\setSetMods$.
\end{theorem}

That is, an adversarial game-master Victor can always choose a distribution $\setData_{\epsilon}$ over words (wheel-spins) such that $\agent_C$ performs no more than $\epsilon$ better than random guessing.

\begin{algorithm}
\caption{Quantum Agent Construction} \label{main:alg:construction}
\SetKwInOut{Input}{input}
\SetKwInOut{Output}{output}
\SetInd{0.5em}{1.25em}
\SetAlgoHangIndent{1.25em}
\Input{prime $p > 2,\setSetMods=\lbrace \setMods_j\rbrace_{j=1}^\intOutputs$}
\Output{$\agent_Q = (\mathcal{S},s_{\mathrm{init}}, \mathcal{T}, \mathcal{X}, \mathcal{Y})$}
\Begin{
    define $\delta_0=0$\;
    set $\delta_k$ to the $k^{\mathrm{th}}$ largest element of $\mathcal{N}_{\setSetMods}$ for k = 1,2,\ldots $N_{\setSetMods}/2$\;
    let $\theta = 2\pi/p$\;
    define ($(|\mathcal{N}_{\setSetMods}|+1) \times p$) matrix $A$ such that \begin{equation*}
    A_{kj}=\begin{cases}
        \cos j\delta_{k'}\theta & k=2k' \textrm{ for } k' \le |\mathcal{N}_{\setSetMods}|/2,\\
        \sin j\delta_{k'}\theta & k=2k'-1 \textrm{ for } k' \le |\mathcal{N}_{\setSetMods}|/2
        \end{cases}
    \end{equation*}
    for $k=0,1,\dots,|\mathcal{N}_{\setSetMods}|$, $k' \in \mathbb{Z}$\;
    let $\mathbf{e}_0 \in \mathbb{R}^{|\mathcal{N}_{\setSetMods}|+1}$ be an $(|\mathcal{N}_{\setSetMods}|+1)$-dimensional vector with $1$ on its first entry and $0$ elsewhere\;
    find nonnegative solution vector $\mathbf{b}\in \mathbb{R}^p$, $\mathbf{b}\ge 0$ to matrix equation
    \begin{equation*}
    A\mathbf{b}=\mathbf{e}_0
    \end{equation*}
    with at most $|\mathcal{N}_{\setSetMods}|+1$ nonzero entries, e.g. using simplex method \cite{stone1991simplex}; by Appendix~\SMconstruct\ this is always possible\;
    let $M \le |\mathcal{N}_{\setSetMods}|+1$ be the number of nonzero entries in $\mathbf{b}$\;
    let $l_i$ be the row-number of the $i$-th nonzero element of $\mathbf{b}$ for $i = 1,2,\ldots, $\;
    let $\mathcal{H}$ be a Hilbert space with an orthonormal basis $\lbrace \ket{\phi_1},\ket{\phi_2},\dots,\ket{\phi_M}\rbrace$\;
    let $\mathcal{S}$ be the density matrices over $\mathcal{H}$\;
    let $ \ket{\sigma_0}=\sum_{i=1}^M \sqrt{b_{l_i}}\ket{\phi_i}$\;
    define $U$ such that $U \ket{\phi_k} = e^{\ii l_k \theta} \ket{\phi_k}$ for $K = 1,\ldots,M$\;
    let $\ket{\sigma_j}=U^j\ket{\sigma_0}$ for $j=0,1,2,\dots,p-1$ \;
    let $\mathcal{M}$ be the measurement according to the observable $\sum_{j=0}^K j P_j$, where $P_j$ is the projection operator onto $S_j=\mathrm{Span}_{k \in \Xi_j}\{\ket{\sigma_k}\}$ for $j>0$, and a rejection operator $P_0 = I - \sum_{j=1}^K P_j$ with output 0 is added for completeness\;
    let $\mathcal X=\mathbb Z_p\cup\{\Omega\}$ and
    $\mathcal Y=\{0,1,\ldots,K\}$\;
    let $s_{\mathrm{init}} = \ket{\sigma_0}\bra{\sigma_0}$\;
    let $\mathcal{T}^x$ for $x \neq \Omega$ be the quantum operation that (i) applies the unitary gate $U^x$ to $S$ and (ii) outputs $y = 0$\;
    let $\mathcal{T}^\Omega$ be the quantum operation that (i) applies $\mathcal{M}$ to $S$, emitting measurement outcome $j \in \mathcal{Y}$ as output, followed by (ii) resetting $S$ to state $s_{init}$\;
}
\end{algorithm}

\textbf{Designing Quantum Agents}. Quantum memory offers an additional opportunity to reduce memory costs without sacrificing performance. To characterize when this is possible, consider the case where,
for some $\Xi_y$ and $\Xi_{y'}$, we have $j\in\Xi_y$ and
$j+\delta\in\Xi_{y'}$ for some $j$ and some $y\neq y'$.
We refer to such a $\delta$ as a \emph{colliding rotation} for
$\setSetMods$. This implies that if we rotate the wheel by $\delta$, then the marble in partition $j$ will change from color $y$ to $y'$. We now define \emph{collision set} $\mathcal{N}_{\setSetMods} = \{(\xi_i - \xi_j) \mod p\,| \, \xi_i \in \setMods_i, \xi_j \in \setMods_j,  \forall \,\,\setMods_i, \setMods_j \in \setSetMods \textrm{ such that }i\neq j\}$ as set of all possible colliding rotations, and the \emph{collision degree} $N_{\setSetMods} = |\mathcal{N}_{\setSetMods}|$ as total number of colliding rotations that a prize-wheel configuration possesses. Our next result shows that the quantum agent constructed using Algorithm~\ref{main:alg:construction} has memory dimension $M = N_{\setSetMods}+1$ (see the proof in Appendix~\SMconstruct).

\begin{theorem}
\label{main:thm:quantumexists}
The quantum agent $\agent_Q$ constructed by Algorithm~\ref{main:alg:construction} solves \taskMSC, and has memory cost at most
\begin{equation}
M = N_{\setSetMods} + 1,
\end{equation}
where $N_{\setSetMods}$ is the collision degree of $\setSetMods = \{\setMods_j\}$.
\end{theorem}

  Thus, $N_{\setSetMods}$ does not necessarily scale with $p$. As illustrative examples, consider the following: (i) a configuration with a blue marble sandwiched between green marbles. This has collision degree $N_{\setSetMods} = 2$, since rotations by $\pm 1 \bmod p$ both lead to a collision. Thus $M_Q \le 3$. (ii) A configuration with three marbles of different colors in sequence (i.e., $\setMods_i = \{i\}$) has collision degree $4$ for any $p \geq 5$, and thus $M_Q \le 5$. Of course, not all configurations exhibit a quantum advantage; for example, this is not the case when there are $K$ different colors with $K > p/2$.

Operationally, the quantum agent $\agent_Q$ has a quantum memory system $S$ whose states lie within an $M = N_{\setSetMods} + 1$ dimensional Hilbert space. We associate each classical state $s_j$ with a corresponding pure memory state $\ket{\sigma_j}$, whose exact form can be determined systematically by Algorithm~\ref{main:alg:construction}. On  input $x \in \mathbb{Z}_p$, applying a unitary $U^x$, as defined in Algorithm~\ref{main:alg:construction} such that
\[
U^x \ket{\sigma_j} = \ket{\sigma_{j+x  \textrm{ mod } p}},
\]
allows the quantum agent to mimic the dynamics of the wheel. On input $\Omega$, Algorithm~\ref{main:alg:construction} then determines a projective measurement on $S$ whose outcome is guaranteed to coincide with the color of the marble to be predicted. The above result therefore establishes that the quantum memory complexity of \taskMSC\ is upper bounded by $N_{\setSetMods}+1$. One might wonder whether one could do better by associating the $s_j$ with mixed states, using non-unitary update rules, or allowing more general measurements. Our final result shows that no such variation can reduce the quantum memory below $N_{\setSetMods}+1$ (Proof in Appendix~\SMunitary\ and \SMoptimal).

\begin{theorem}
\label{thm:quantumlowerbound}
Any quantum agent $\agent_Q$ solving \taskMSC\ with zero loss has memory cost at least $M = N_{\setSetMods}+1$. In particular, the agent constructed by Algorithm~\ref{main:alg:construction} is memory-optimal.
\end{theorem}

Therefore, we have $M_Q = N_{\setSetMods}+1$. This immediately enables us to construct a plethora of different families of classification problems in which separation between classical and quantum memory complexity can scale without bound. For example, consider the subclass of classification problems BPG1($p,K$)=\taskMSC\ with $\setMods_i = \{i\}$, corresponding to the case where we have $K$ different color marbles initialized in sequence. Then $M_Q = 2K - 1$ for sufficiently large $p$ while $M_C = p$. As such, for each fixed $K$, we have a family of $K$-class classification games where quantum advantage grows without bound with increasing $p$.

\section{Discussion} Memory is a key resource in online classification, enabling agents to correctly label sequentially observed data even when the correct label depends critically on data from the distant past. Here, we introduced a collection of such classification tasks - including those with \(K\) classes for any \(K \geq 2\) - involving queries about a dynamically evolving system that can be in \(p\) possible configurations. We showed that, for prime \(p\), perfect classical play requires a memory dimension of \(p\), corresponding to exact tracking of the system's configuration. In contrast, for each fixed \(K\), we constructed quantum agents that achieve perfect classification with memory costs bounded independently of \(p\). This separation is sharp: below the classical memory threshold, suitably chosen input distributions force any classical agent to perform arbitrarily close to random guessing. Furthermore, we show that our quantum constructions are memory-minimal, and thereby establish an exact separation between classical and quantum memory complexity that can grow without bound. Our results therefore identify a new route to quantum advantage during model inference, complementary to existing computational and memory advantages in model discovery~\cite{rebensvm2014,zhao2026exponential}.

A natural next direction is to identify commonalities between quantum memory advantage in online classification and analogous advantages in other sequential tasks~\cite{ambainis98,kallaugher2025design,vieira2022temporal,Chengran2023,yang2025dimension,lumbreras2026irreduciblequantumadvantagealigning,sundar2026dimensionreductionquantumadaptive}. For example, our classification games with \(K=2\) labels encompass previously studied promise problems exhibiting a quantum memory advantage~\cite{tian2019experimental}, suggesting that quantum model optimality results could translate to such settings. Meanwhile, memory advantage has been shown to imply energetic advantage in the context of quantum stochastic models and quantum agents for replicating specific input-output behaviors~\cite{loomis2020thermal,thompson2025energetic}. Could such memory advantages also yield a quantum energetic advantage in classification? More generally, one may expect more favorable memory--performance tradeoffs once perfect play is relaxed, as already observed in timekeeping and stochastic modeling~\cite{Chengran2023,woodsclock,yang2025dimension}.


Learnability is another compelling continuation: the quantum models introduced here are hand-built, but can they also be systematically discovered? Here, memory-constrained model classes may be especially valuable, since they enjoy tighter generalization bounds, a phenomenon that persists even when data are encoded quantum mechanically~\cite{banchi2021generalization}. Our quantum memory advantage therefore suggests the possibility of a corresponding learning advantage, potentially mitigating a major concern about single-shot machine learning in unconstrained-memory settings~\cite{recio2025single}. Should we formally establish such an advantage, it would represent an exciting pathway for quantum machine-learning advantage in data-scarce environments.

\section*{Acknowledgements}
This work is supported by the National Research Foundation of Singapore through the NRF Investigatorship Program (Award No. NRF-NRFI09-0010), the National Quantum Office, hosted in A*STAR, under its Centre for Quantum Technologies Funding Initiative (S24Q2d0009), the RIE 2025 AQAS projects S25Q9D001 and S25Q9D002, the Singapore Ministry of Education Tier 1 Grant RT4/23 and RG91/25 and the RIE25 Japan-Singapore Joint Call on Quantum (Project ID H25-MRO3490).

\emph{AI Disclosure}. The authors derived all technical results and produced main manuscript drafts. OpenAI’s ChatGPT and Codex were used to assist with
proofreading, checking for notational consistency, and making the prize wheel in Figure 2 look more attractive. All content was reviewed and verified by the authors.

\bibliography{mile_promise}

\begin{appendix}
\section{Time Series Classification}

Consider an agent that is interacting with its environment. At each point in time $t \in \mathbb{N}$ the agent receives a character $\charInput_t \in \mathcal{X}$. The agent continues to collect inputs over multiple consecutive time steps, until it receives a  termination character $\charInput_\intLen = \Omega$.

We define the resulting sequence of inputs the agent has received as a string $\strInput=[\charInput_0, \charInput_1, \charInput_2, \ldots, x_{\intLen-1}]$. Strings can be of arbitrary length, so there is no way to anticipate in advance when the termination character $\charInput_T = \Omega$ will be received.

In general we expect the agent to encounter this data in an  online streaming format. Such that the agent receives inputs $x_t \in \mathcal{X}$ over sequential points in time $t \in \mathbb{Z}$, culminating in a termination character which delineates a word $\strInput^{(i)}$, and immediately on the next time step the agent receives the first character of the next string $\strInput^{(i+1)}$. Over time this results in a data stream:
\begin{equation}
    \stream = \strInput^{(0)} \Omega\strInput^{(1)}\Omega\ldots \Omega\strInput^{(k-1)}\Omega\ldots
\end{equation}
where we refer to $\strInput^{(k-1)}$ as the $k^{th}$ string.

Each string $\strInput$ in the data stream, {\it can have} an associated classification label $\fnClass{\strInput}.$ For convenience, we will also define the \textit{class} of a word as the set of strings with the same classification, i.e.
\begin{equation}
    \setClass_j = \left\{\strInput \, \middle| \,\fnClass{\strInput}=j\right\}.
\end{equation}
If a string $\strInput \in \Gamma_j$, then we say the `correct' classification label $\fnClass{\strInput} = j$.  We assume there are $K$ classes, such that the number of classes $\setClasses= \{\Gamma_j\}_{j= 1}^K$ is finite. If a string $\strInput \notin \cup_j \Gamma_j$ does not belong to any class, then we say it does not have a correct classification label.  We thus consider $\fnClass{\strInput}$ as a \textit{partial} map from inputs to outputs.

The problem of Time Series Classification (\taskTSC) can then be defined as follows:
\begin{problem}
 \label{def:TSC} {\bf Time Series Classification (\taskTSC):}
Consider a data stream $ \overrightarrow{w} $ comprised of strings $\strInput$, with `correct' classification labels $\fnClass{\strInput}$ whenever $\strInput \in \cup_j \Gamma_j$.
The agent receives an input character $\charInput_t$, at each sequential point in time $t\in \mathbb{N}$, until it sees the termination character $\charInput_\intLen = \Omega$. At this timestep, the $\agent$ must output a guess $y_T =  \fnGuess{\strInput}$, such that it obtains zero loss under the loss function
\begin{equation}
\fnLoss{\fnGuess{\strInput},\fnClass{\strInput}} = \begin{cases}
0 & \text{if }  \fnGuess{\strInput}= \fnClass{\strInput}  \text{ or } \strInput \notin \bigcup_j \Gamma_j \\
1 & \text{otherwise}
\end{cases}
\end{equation}
\end{problem}

Our loss function is chosen to award a loss of 1  if the classification is incorrect, and 0 otherwise.  We say the agent solves \taskTSC \, if it always obtains zero loss. We play this game on repeat requiring the agent to classify every word in the data stream correctly.

Supposing the agent cannot solve \taskTSC, then there will be some words for which it incurs a non-zero loss. To benchmark such an agent, a referee will choose strings in the data stream from some distribution $\setData$, and score the agent based on its loss for the tasks \taskTSC. The expected loss per word $\fnELoss{\agent}$ over $\setData$, is then equal to the probability of misclassification:

\begin{equation}
\fnELoss{\agent} = \sum_{{\strInput} \sim \setData}  \setData({\strInput}) \sum_{y\in\mathcal Y}
P(\fnGuess{\strInput} = y)
L\!\left(y,\fnClass{\strInput}\right)
\end{equation}

In this paper, we will assume that the distribution $\setData$ is \textit{not known a priori}. An agent cannot in general use foreknowledge of $\setData$, to optimize its construction.

\begin{figure*}[hbtp]
    \centering
    \includegraphics[width=0.9\textwidth]{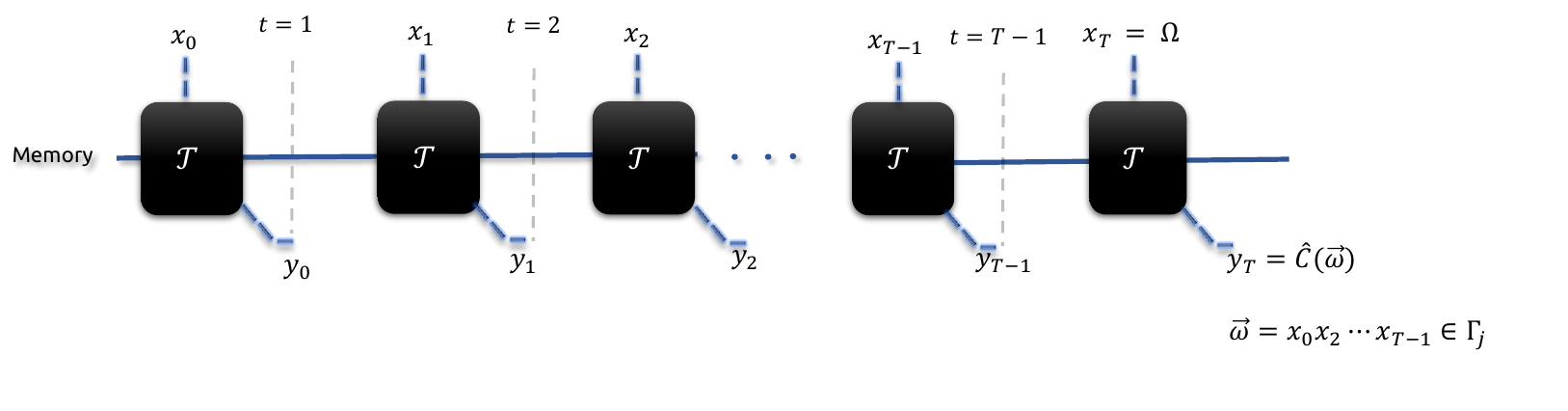}
    \caption{At each time step $t$ an agent receives an input symbol $x_t$, drawn from a stream  $\strInput= x_0 x_1 \ldots x_{\intLen-1}$. This stream is terminated by an $\Omega$ at time $t=\intLen$.  Upon receiving the termination input, the agent must respond by (probabilistically) emitting the label $y_{\intLen} = j$ corresponding to the word's classification.
    }
    \label{fig:transducer}
\end{figure*}

\section{Classical and Quantum Agents}

Suppose we designed an agent to solve this problem. This agent (shown in Fig. \ref{fig:transducer}) would accept input $\charInput_t$ at every time-step, storing information about previous inputs in memory. We will assume the agent can also emit an output $\charOutput_t$ at every timestep, although $\charOutput_t$ {\it will not be used to evaluate the loss function whenever} $\charInput_t \neq \Omega$. On receiving a termination character $\charInput_{\intLen} = \Omega$, the agent classifies the string of received inputs $\strInput$ into one of $K$ possible classes, outputting its guess as $\charOutput_\intLen = \fnGuess{\strInput} $.

Let us now define our agent model.
\begin{definition}
A \textit{classifying agent} $\mathcal{A}$
is defined by a tuple
\begin{equation}
\label{eq:agentgentuple}
    \mathcal{A} = \tupleAgent
\end{equation}
where
\begin{itemize}
    \item $\setInputs = \{x\}$ is the set of inputs it may accept,
    \item $\setOutputs =\{y\} \supseteq \{1,\ldots, K\} $ is the set of outputs it may emit,
    \item $\setpStates \subseteq D(\mathcal{H})$ is the set of memory states the agent may store in its memory, where $D(\mathcal{H})$ denotes the density operators on a finite-dimensional Hilbert space $\mathcal{H}$, which we call the agent's \textit{memory space},
    \item $\causalstate_{init}\in \setpStates$ is the initial state of the agent's memory (here often set as $s_0$ to simplify notation), and
    \item $\fnTransition$ specifies a policy map that acts jointly on the Hilbert space encoding the input $x \in \mathcal{X}$ and $\mathcal{H}$. The policy specifies how the agent selects the output $y \in \mathcal{Y}$ and updates the state of the memory $\mathcal{S}$.
    \end{itemize}
\end{definition}

At every timestep, the agent's memory state \(s_t\) is represented by a
density operator on the memory space $\mathcal{H}$. To  represent this state we introduce a set of pure states $\setStatesBasis=\{\sigma_i\}$.

When the agent is classical, we enforce these pure states to form an orthogonal basis     $\sigma_i:=\ket{i}\!\bra{i}$ such that $\{\langle  i |j\rangle  = \delta_{ij}, \,\,  \forall \sigma_i, \sigma_j \in \setStatesBasis\}$
\[
    \mathcal{S}_C
    =
    \left\{
    s=\sum_i p_i \sigma_i
    \;:\;
    p_i\ge 0,\ \sum_i p_i=1
    \right\}.
\]
Thus, in each individual run, the classical machine occupies one definite
configuration \(\sigma_i\), while the distribution \(\{p_i\}\) represents
our uncertainty about the current state, based on having seen the past history of inputs and outputs. Equivalently $p_i$ is the probability the agent will be in state $\sigma_i$. Stochastic
state updates can therefore produce nonzero probability of being in multiple different
classical configurations, but never coherent superpositions between them.

In contrast to the classical case, quantum agents can store and manipulate information in quantum memory, and thus we take the quantum agent's memory states to be general density matrices over the memory space $\mathcal{H}$.

\begin{definition} \label{def:memory}
We define the memory cost $M$ of an agent $\mathcal{A}=(\mathcal{X},\mathcal{Y},\mathcal{S}, s_{init} = s_0,\mathcal{T})$ as the smallest dimension of any physical system that can store  $\mathcal{S}$ and recover the states with 100\% accuracy by a reversible transformation (e.g. unitary).
\end{definition}

In the literature \cite{ barnett2015computational, elliott2022} this is sometimes referred to as the \textit{topological complexity}. Or simply the dimension of the smallest Hilbert space $\mathcal{H}$ which encodes the memory states. Classically it aligns with the number of states $\sigma_i = \ket{i}\bra{i}$ in $S_b$, so classically $M = |\mathcal{S}_b|$. Our goal is to find the agent $\mathcal{A}$ solving a classification problem with the lowest memory cost possible.

 We will talk about this smallest Hilbert space which stores all the agent's memory states as the accessible Hilbert space  -denoted $\mathcal{H}(\agent)$ or $\mathcal{H}$ (when it is unambiguous which construction it refers to).   We will assume henceforth that the Hilbert space we are dealing with is this accessible Hilbert space.

\begin{definition}
An \textit{optimal} agent $\agent = \tupleAgent$ solving \taskTSC \, has memory cost $M$ no larger than that of any other agent solving \taskTSC; that is, for any such agent $\agent'$ with memory cost $M',$ $M \le M'.$
\end{definition}
When clear from context, we will often simply call such an $\agent$ an `optimal agent'.


In order to correctly classify a given stream of inputs, the agent requires a dynamical map. At each time step upon receiving input $x_t \in \mathcal{X}$, execution of this map generates the output $y_t \in \cal{Y}$ and simultaneously updates the agent's memory register accordingly by an update map $\mathcal{T}^{y_t|x_t}$. This process is described via a completely positive trace preserving map (CPTP map)
\begin{equation}\label{eq:krauschanel}
    \mathcal{T}: \ket{x} \bra{x} \otimes s_i \rightarrow \sum_{y}  \ket{y}\bra{y}\otimes \mathcal{T}^{\charOutput| \charInput} (s_i)
\end{equation}
In the case of a classical agent, for each possible $x$ this map is a substochastic nonnegative transition matrix  $\mathcal{T}^{\charOutput| \charInput} (\sigma_i) = \sum_k P(y, \sigma_k | x, \sigma_i) \sigma_k $ which describes the Markov chain dynamics by which an agent initially in state $\sigma_t$ generates the output $y_t$ and next state $\sigma_{t+1}$. Note that as this map is linear, its action on classical states $s_i = \sum_j p^i_j \sigma_j$ is implicitly defined by its action on the basis $\setStatesBasis= \{\sigma_i\}$.

For quantum agents, the map can be described by a collection of Kraus operators, corresponding to a completely positive trace preserving channel. Suppose, on input $x$ and internal state $s$, our agent outputs $y$, then in Eq. \ref{eq:krauschanel} we write
\begin{equation}
\label{eq:tyx}
    T^{(y|x)}(s_i)=\sum_j K^{(y|x)}_j s_i K^{(y|x) \dagger}_j
\end{equation}
where $K^{y|x}_j$ are the Kraus operators associated with $\mathcal{T}$, conditioned on being given input $\charInput$ and output $\charOutput$. In particular for each $\charInput \in \mathcal{X}$ we have a completeness relation $\sum_y \sum_j {K^{(y|x)}_j}^{\dagger}  K^{(y|x)}_j  = \openone$. Meanwhile  we interpret the quantum map transitions analogously  to the classical through defining ${\rm Tr}\left[  \sum_j K^{(y|x)}_j s_i K^{(y|x) \dagger}_j\right] := P(y | x, s_i)$ and the resulting conditional state of the memory register as  $s' := \left[\sum_j K^{(y|x)}_j s_i K^{(y|x) \dagger}_j\right]/P(y| x, s_i)$.

\section{Black box Prize Wheel and Counting Problems}

We now define the problem that is the focus of this paper. Fix $p$ prime, and let $\setInputs = \setZp \cup \{\Omega\}$. Throughout the remainder of the paper, unless explicitly stated
otherwise, we take $p$ to be an odd prime (the $p =2$ case is fundamentally uninteresting, as both quantum and classical agent's need at least 2 internal states). Each string in the data stream $\strInput = [x_0, x_1, \ldots, x_{\intLen-1}]$  is thus comprised of integers modulo $p$.


Let $\setSetMods = \{\setMods_j\}$  be a collection of disjoint subsets $\setMods_j \subseteq \setZp.$ Our classification classes then can be defined in terms of whether the sum of the characters in the string falls in one of these subsets $\setMods_j$. In particular we define the classes implicitly by

\begin{equation}\label{sq:classesMSC}
    \strInput \in \setClass_j  \qquad \textit{iff} \qquad\left(\sum_{i=0}^{T-1} \charInput_i \right)~\mod p \in \setMods_j.
\end{equation}

An agent who solves $\taskTSC$ (time serries classification) based on these classes, is computing which subset $\setMods_j$ the string remainder belongs to. The relations it computes are captured by Fig. \ref{fig:clock}.

\begin{problem}
\label{prob:summation}
{\bf Black box prize wheel problem \taskMSC:}
Fix $p$ prime. Consider an instance of  where  $\setInputs = \setZp \cup \{\Omega\}$,  and  classes  $\setClasses = \{\Gamma_1,\ldots, \Gamma_K\}$ are implicitly defined by  Eq. \eqref{sq:classesMSC} with disjoint subsets  $\setSetMods = \{\setMods_1, \setMods_2,\dots,\setMods_K\}$ of  $\setZp$.

We call the resulting time series classification task $\taskMSC$: black box prize wheel problem mod $p$ over $\setSetMods$.
We say a classifying agent  $\agent$ \textit{solves} $\taskMSC$  if it obtains zero loss on all words under the 0-1 loss function in Problem \ref{def:TSC}.
\end{problem}

\begin{figure}[hbtp]
    \centering
    \includegraphics[trim={8cm 0cm 8cm 0cm}, clip, width=0.9\columnwidth]{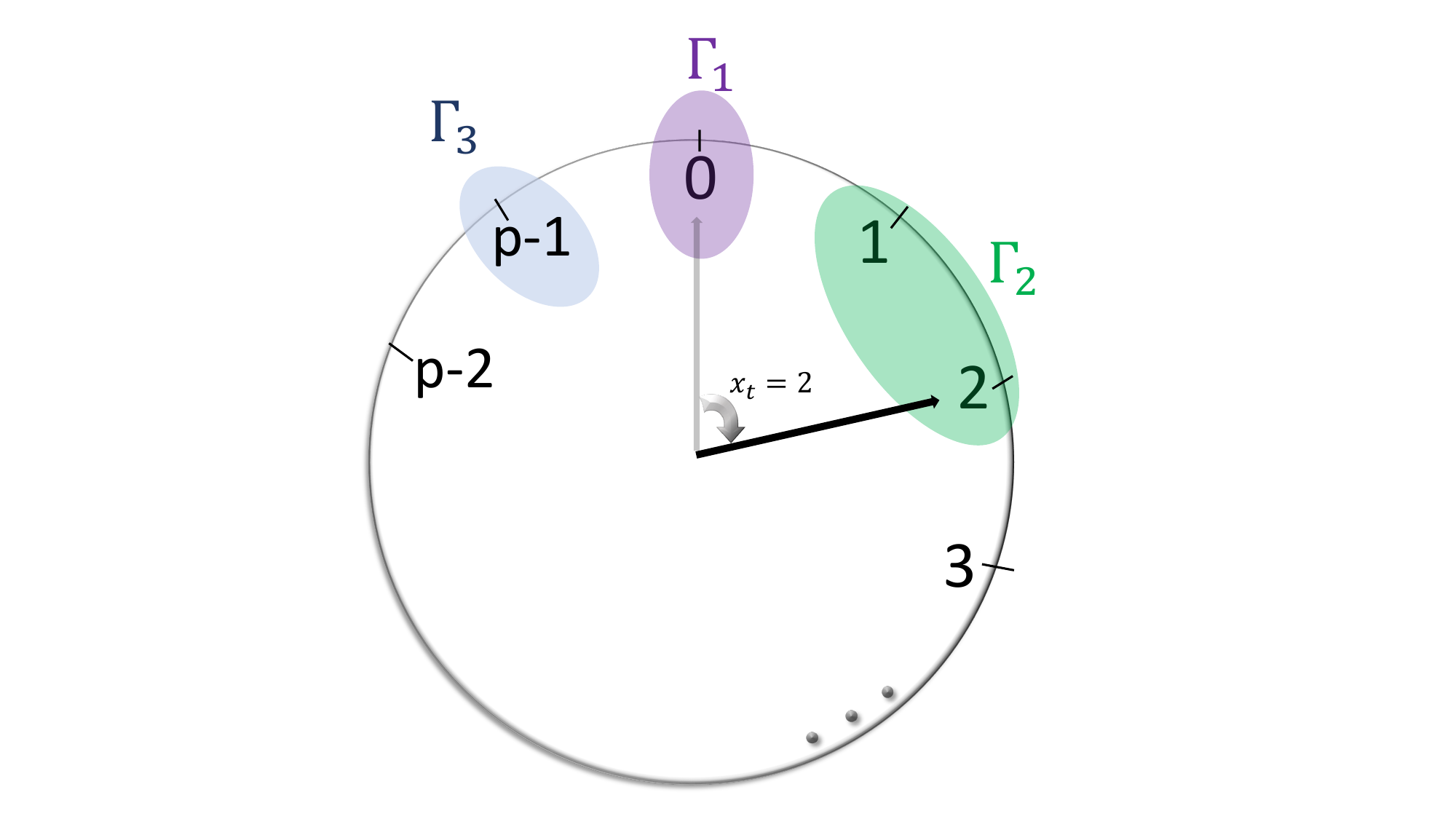}
    \caption{An example instance of a \taskMSC \, problem for some prime $p$. An agent must keep track of where on the prize wheel for the current input word $\strInput$ -- as illustrated by the effect of receiving an input of $x_t = 2$ on top of a running total which is congruent to $0 \textrm{ mod } p$. Upon receiving the termination symbol $\Omega$ on (\textit{a priori} unknown) timestep $\intLen$, the agent needs to output the label $j$ of the class $\setClass_j$ such that $\sum_{t=0}^{T-1} \charInput_t \textrm{ mod } p \in \setMods_{j}$.  If the word does not belong to any class (i.e. $(\vec{w}\notin \bigcup_j \Gamma_j$)), then there is no promise on the output behaviour. In this particular instance, while any classical agent must have memory cost $p$, a quantum agent can solve this problem with memory cost $7$, \textit{regardless} of how large $p$ is. See Section \ref{sec:construction} for details.}
    \label{fig:clock}
\end{figure}

We also show that, to solve this problem, it is sufficient to solve a simpler version of this problem where $\setInputs=\{0, 1, \Omega\}$; in this case, the class of each timeseries $\strInput$ is determined by its sum.

\begin{problem}
\label{prob:counting}
{\bf Modular Counting Classification Problem $\taskMCC$:}
Fix $p$ prime.  Consider an instance $\taskTSC$  where $\setInputs = \{0, 1,\Omega\}$, and  classes  $\setClasses = \{\Gamma_1,\ldots, \Gamma_K\}$ are implicitly defined by  Eq. \eqref{sq:classesMSC} with disjoint subsets  $\setSetMods = \{\setMods_1, \setMods_2,\dots,\setMods_K\}$ of  $\setZp$.

We call the resulting time series classification task $\taskMCC$: the modular counting classification problem mod $p$ over $\setSetMods$.
We say a classifying agent  $\agent$ \textit{solves} $\taskMCC$  if it obtains zero loss on all words under the 0-1 loss function in Problem \ref{def:TSC}.
\end{problem}

\section{Simple Results on General Classification Problems}
At this point, we present a few simple lemmas regarding the implementation of optimal (i.e. minimal memory dimension agents that solve $\taskTSC$\, with zero loss) quantum and classical agents $\agent = \tupleAgent.$ \eqref{eq:agentgentuple}. In this discussion, an ``optimal classical agent'' is a classical agent which solves $\taskTSC$ with zero loss, and with minimal memory among all such \textit{classical} agents. By previous definitions, of course, any classical agent is also a quantum agent, and so optimal quantum agents must use no more memory than optimal classical agents. We can also prove some simple assertions about agents solving $\taskTSC$.

\begin{lemma}
\label{lemma:reset}
If a classical or quantum agent $\agent$ has zero loss for the first string $\strInput^{(0)}$ in any data stream $ \stream$, then it can be used to solve $\taskTSC$. Furthermore, this choice does not increase the memory cost  $M =\dimHacc$ of the agent  $\agent$.
\end{lemma}
\begin{proof}
The first statement follows trivially from the existence of an initial state $s_{init}=s_0$. Upon receiving termination character $\charInput_\intLen = \Omega$ and outputting the classification label $\charOutput_\intLen = \hat{C}(\strInput)$, the agent can reset its internal state to $s_{\intLen+1} = s_0$. The resulting agent will  classify strings $\strInput^{(k)}$ in a way that is independent of $k$ (ensuring it is correct on all strings in $\stream$). Furthermore, resetting the agent ensures it can reuse the same state space on every string $\strInput^{(k)}$. This can not increase the total memory dimension $M$.
\end{proof}

We can now easily show that there exists an optimal agent which outputs $\charOutput_t = 0$ whenever $\charInput_t \neq \Omega$.

\begin{lemma}
\label{lemma:zero}
For any quantum agent (or classical agent) $\agent'$ there exists another agent $\agent$ which always responds to any input $\charInput_\intLen \neq \Omega$ with 0, resets the internal state after receiving $\Omega$, has the same loss on every word for \taskTSC \ when initialized in $s_0$, and uses no more memory. In particular, there exists an \textit{optimal} agent $\agent$ which behaves in this way.
\end{lemma}

\begin{proof}
Consider an agent $\tilde{\mathcal{A}} = (\setInputs, \setOutputs, \setpStates, s_0, \tilde{\mathcal{T}})$ which does not necessarily respond to normal input $x_t \neq \Omega$ with $0$. Because $\tilde{\mathcal{A}}$ is a quantum agent, its transition map $\tilde{\mathcal{T}}$ is completely positive and trace-preserving (CPTP), and can be described by a set of Kraus operators $\{\tilde{K}_i^{(y|x)}\}$ acting on the memory space $\mathcal{H}$.

We construct a new agent $\mathcal{A} = (\setInputs, \setOutputs,  \setpStates, s_0, \mathcal{T})$ designed to force the output $y_t = 0$ for all normal inputs. For any input $x \neq \Omega$, we define the new transition map by lumping all possible original outcomes $y$ into the single output label $0$. We achieve this by replacing the environment index $i$ with a joint index $(i, y)$:
$$K_{(i,y)}^{(0|x)} = \tilde{K}_i^{(y|x)}$$
For any output $y \neq 0$, the set of Kraus operators is empty ($K_k^{(y \neq 0 | x)} = 0$).

Because this new map simply relabels the classical output without altering the physical operators acting on the quantum memory, the unconditioned state evolution of the memory remains mathematically identical:
$$\mathcal{E}^{(x)}(\rho) = \sum_{y, i} K_{(i,y)}^{(0|x)} \rho K_{(i,y)}^{(0|x)\dagger} = \sum_{y, i} \tilde{K}_i^{(y|x)} \rho \tilde{K}_i^{(y|x)\dagger}$$
Using Eq.\eqref{eq:support}, the subspace spanned by the accessible memory states of $\mathcal{A}$ is exactly the same as that of $\tilde{\mathcal{A}}$. Therefore, the memory cost of the agent $M = \dimHacc$ has not increased.

For the final timestep where $x_T = \Omega$, we leave the measurement operators unchanged: $K_i^{(j|\Omega)} = \tilde{K}_i^{(j|\Omega)}$. Because the intermediate memory states leading up to the termination character are identical, and the final measurement is identical, the probability of outputting the correct classification $\hat{C}(\strInput) = j$ remains exactly as before, and so the new agent achieves the same loss.

Finally, to satisfy the reset condition, we post-process the state after the measurement at $x_T = \Omega$. We apply a standard CPTP map that traces out the post-measurement internal state and prepares the initial state $s_0$. Because $s_0$ is, by definition, already an accessible state $\mathcal{H}(\agent)$, this reset operation strictly cannot increase the dimension of the subspace.

Thus an agent satisfying these strict output and reset conditions exists
with the same loss as $\tilde{\agent}$ and no greater memory cost. In particular, this is true when $\tilde{\agent}$ is an \textit{optimal} agent; thus an optimal agent $\agent$ satisfying these strict requirements always exists.
\end{proof}

Next we would like to prove the initial state can always be chosen to be pure. Before doing this we establish some useful facts about CPTP maps. First write the initial state of the quantum agent's memory as a convex combination of pure states:
\begin{equation}
\label{eq:initial-convex-decomposition}
\rho_0=\sum_a q_a \proj{\psi_a}{\psi_a},
\qquad q_a>0,\quad \sum_a q_a=1 .
\end{equation}
We then observe the following support identity. For any
CPTP map \(\mathcal E\),
\begin{equation}\label{eq:support}
\operatorname{supp}\!\left(\mathcal E(\rho_0)\right)
=
\operatorname{span}\left(
\bigcup_a \operatorname{supp}\!\left(\mathcal E(|\psi_a\rangle
\!\langle\psi_a|)\right)
\right).
\end{equation}
A simple way to prove this equation is as follows. Let $X_a = \mathcal{E}(\proj{\psi_a}{\psi_a});$ by physicality $X_a$ is positive semidefinite.
But for \textit{any} positive semidefinite operators $X_a$,
\begin{equation}
\ker\!\left(\sum_a q_a X_a\right)=\bigcap_a \ker(q_a X_a),
\end{equation}
which can be seen directly by the fact that $X_a$ are positive semi definite and $q_a > 0$, and so $\ematrix{\gamma}{\sum_a q_a X_a}{\gamma}  = 0 $ \textit{iff} $ q_a\ematrix{\gamma}{X_a}{\gamma}=0$ for all $a$. Taking orthogonal complements gives the support identity Eq. \eqref{eq:support}.

\begin{lemma}
\label{lemma:purity}
Given any agent $\agent'=(\mathcal{X},\mathcal{Y},\mathcal{S},\rho_0,\mathcal{T})$, there exists some pure state $s_{init} = s_0$ in the support of $\rho_0$ such that the pure-initial-state agent $\agent=\tupleAgent$ has the same memory cost or better, \textit{and} has the same expected loss or better. In particular, there exists an optimal quantum (or classical) agent $\agent$ solving  $\taskTSC$ that is initialized in a pure state $s_0$.
\end{lemma}

\begin{proof}
To begin, by Lemma \ref{lemma:zero}, we may assume the agent $\agent'$  outputs $y =0$ on every
non-terminating input $\charInput_t \neq \Omega$.

Let $T^{(y|x)}$ be as in Eq.\eqref{eq:tyx}. We then define, for each finite word $\vec w=x_0\cdots x_{T-1}$,
\begin{equation}
  T^{(\strInput)}
=
T^{(0|x_{T-1})}\cdots T^{(0|x_0)};
\end{equation}
in particular, $T^{(\strInput)}(\rho_0)$ is precisely the memory state of the agent after input string $\vec w.$

Now, let us write the initial state of the quantum agent's memory as a convex combination of pure states,
\begin{equation}
\label{eq:convexcombination}
\rho_0=\sum_a q_a \proj{\psi_a}{\psi_a},
\qquad q_a>0,\quad \sum_a q_a=1 .
\end{equation}
If the final
measurement is described by positive semidefinite operators $\{E_j\}_j$, then the probability of
classifying $\vec w$ into class $j$ is
\begin{equation}
p(j|\vec w,\rho_0)
=
\operatorname{Tr}\!\left[E_j\, T^{(\strInput)}(\rho_0)\right].
\end{equation}

As this expression is linear in the initial state, we can equivalently re-express this probability as
\begin{equation}
p(j|\vec w,\rho_0)
= \sum_a q_a p(j|\vec w,\proj{\psi_a}{\psi_a})
\end{equation}
where  $p(j|\vec w,\proj{\psi_a}{\psi_a})=
\operatorname{Tr}\!\left[E_j\, T^{(\strInput)}(|\psi_a\rangle\!\langle\psi_a|)\right]$ is the probability an agent initially in pure state $\ket{\psi_a}\bra{\psi_a}$ emits the correct label $j$ upon input  $\vec w \in \Gamma_j$. (More precisely, such an agent $\mathcal{A}''=(\mathcal{S}, \proj{\psi_a}{\psi_a}, \mathcal{T}, \mathcal{X}, \mathcal{Y})$ is identical to $\mathcal{A}'$, except for the altered initial state.)

Thus, the loss for input word $\vec w \in \Gamma_j$ is:
\begin{align}\label{eq:affineloss}
L_{\vec w}(\rho_0)
&=
1-p(j|\vec w,\rho_0)\nonumber\\
&=
\sum_a q_a\left[1-p(j|\vec w,\proj{\psi_a}{\psi_a})\right].
\end{align}

For a fixed distribution $\setData$, define the expected loss of the
agent initialized in state $\rho$ (again, with all other parameters equal) by
\begin{equation}
\mathcal{L}_{\setData} (\rho)=
\sum_{\vec w\sim \setData}
\setData(\vec w)L_{\vec w}(\rho).
\end{equation}

By Eq.\eqref{eq:affineloss},
\begin{equation}
\mathcal{L}_D(\rho_0)=\sum_a q_a \mathcal{L}_D(\proj{\psi_a}{\psi_a});
\end{equation}
that is, the expected loss for an agent initialized in $\rho_0$ is the convex combination of the expected losses of agents initialized in the support states $\proj{\psi_a}{\psi_a}$
Therefore at least one \(a^\star\) satisfies
\begin{equation}
\mathcal{L}_D(|\psi_{a^\star}\rangle\langle\psi_{a^\star}|)
\le \mathcal{L}_D(\rho_0);
\end{equation}
choosing $s_0=|\psi_{a^\star}\rangle\langle\psi_{a^\star}|$, we find an agent $\agent=\tupleAgent$ initialized in a pure state with expected loss no worse than that of $\agent'.$

In the special case where $\agent'$ solves TSC with zero loss, this also implies that for every $a$ with
$q_a>0$, and every promised word $\vec w$,
\begin{equation}
L_{\vec w}(\proj{\psi_a}{\psi_a})=0.
\end{equation}
Thus, for \textit{any} pure initial state $s_0=\proj{\psi_a}{\psi_a}$ in the support of $\rho_0$, $\agent=\tupleAgent$ also solves $\taskTSC$ with zero loss.

More generally, for any normalized
$\ket{\phi}\in\operatorname{supp}(\rho_0)$, one may choose $q>0$
sufficiently small that
\begin{equation}
\rho_0
=
q\ket{\phi}\!\bra{\phi}
+
(1-q)\tau
\end{equation}
for some density operator $\tau$. Applying the same nonnegativity
argument to this decomposition shows that
$L_{\strInput}(\ket{\phi}\!\bra{\phi})=0$ for every promised word $\strInput$. Finally, in the case of a classical agent, we may always choose the decomposition of $\rho_0$ into
computational-basis projectors; with this choice, the pure initial state is also a valid classical memory state.

Thus, if we have an agent $\agent'$ with initial state $s_0' = \rho_0$, then we can always replace it with a counterpart $\agent$ whose initial state is $s_0 = \proj{\psi_a}{\psi_a}$ and the loss for $\taskTSC$\, will be no worse. It remains to show that this new agent $\agent$ will have memory cost no greater than that of $\agent'$.

However, this is a simple result of the support identity Eq.\eqref{eq:support}.
Thus it follows that if we consider the span of the memory subspace over all
reachable states for agent $\agent$, then this is subspace of the equivalent reachable subspace of $\agent'$. Hence memory cost of $\agent$  can not exceed that of $\agent'$.
\end{proof}

Next, we prove a result regarding the final measurement.
\begin{lemma}
\label{lemma:projective}
There exists an optimal quantum agent $\agent = \tupleAgent$ which applies a \textit{projective} measurement upon receiving input $\charInput_\intLen = \Omega$ and outputs the result (and then resets its memory).
\end{lemma}
\begin{proof}

Consider two input strings $\strInput^{(1)}, \strInput^{(2)}$ with distinct classifications $\fnClass{\strInput^{(1)}}\neq \fnClass{\strInput^{(2)}}$.
Applying Lemma \ref{lemma:zero}, let us denote the \textit{update map induced by $\strInput$} as
\begin{equation}
    T^{(\strInput)}(s) = T^{(0|x_{\intLen-1})}\dots T^{(0|x_1)}T^{(0|x_0)}(s);
\end{equation}
in particular, the internal state induced by $\strInput^{(1)}$ (or $\strInput^{(2)}$) immediately before $\Omega$ is $T^{(\strInput^{(1)})}(s_0)$ (respectively $T^{(\strInput^{(2)})}(s_0)$).

Consider an optimal agent $\agent'$ which does \textit{not} use a projective measurement.
Now, because $\agent'$ classifies with no loss, it must perfectly distinguish $T^{(\strInput^{(1)})}(s_0)$ from $T^{(\strInput^{(2)})}(s_0)$. Thus, by quantum state discrimination, the supports of these states must be orthogonal to each other.

For $\strInput \in\Gamma_j$ and $\strInput'\in\Gamma_{j'}$, with
$j\neq j'$, write
\begin{equation}
\rho_{\strInput}= T^{(\strInput)}(s_0),
\qquad
\rho_{\strInput'}= T^{(\strInput')}(s_0).
\end{equation}
If $\{E_k\}_k$ is the original zero-error POVM, then
\begin{equation}
\operatorname{Tr}(E_j\rho_{\strInput})=1,
\qquad
\operatorname{Tr}(E_j\rho_{\strInput'})=0.
\end{equation}

Since $0\leq E_j\leq\mathbbm 1$, the first equality implies that
$E_j$ acts as the identity on $\operatorname{supp}(\rho_{\strInput})$,
while the second implies that it vanishes on
$\mathrm{supp}(\rho_{\strInput'})$. Hence
\begin{equation}
\mathrm{supp}(\rho_{\strInput})
\perp
\mathrm{supp}(\rho_{\strInput'}).
\end{equation}

For each class $j$, define
\begin{equation}
S_j:=
\mathrm{span}\!\left(
\bigcup_{\strInput \in\Gamma_j}
\mathrm{supp}(\rho_{\strInput})
\right).
\end{equation}

Taking spans over all words in the two classes gives
$S_j$ perpendicular to $S_{j'}$.

Let $P_j$ be the projector onto $S_j$ and set
\begin{equation}
P_{\perp}=\openone-\sum_j P_j.
\end{equation}
We can then (for example) define
\begin{equation}
\Pi_1=P_1+P_\perp,
\qquad
\Pi_j=P_j\quad (j\geq2).
\end{equation}
Then $\{\Pi_j\}_j$ is a projective measurement and is correct on
every promised word.

 Since this measurement does not affect the space of accessible states at all, an agent $\agent$ using this new measurement uses the same amount of memory as the previous agent, proving the lemma.

\end{proof}

\section{Reducing Black Box Prize Wheel game to Modular Counting}

Specializing to $\taskMSC$ and $\taskMCC$, we shorten our definition of an Agent to $\mathcal{A}=(\mathcal{S},s_{init} = s_0,\mathcal{T})$, as $\mathcal{X} = \mathbb{Z}_p \cup \{\Omega\},   \mathcal{Y} = \{0, 1,\ldots,K\}$ are fixed by the problem specification.

First note that by Lemma \ref{lemma:zero}  when $x \neq \Omega$ we only need to worry abut the behaviour $\fnTransition'^{(y =0|x)}$ and $\fnTransition^{(y =0|x)}$. We start by showing the following:
\begin{lemma}
\label{lemma:tally2}
Suppose the agent $\agent=(\mathcal{S},s_{init} = s_0,\mathcal{T})$ solves \taskMSC \, and assume also that $ \fnTransition^{(y|x)}=0$ for all cases where: $x \neq \Omega$ and $y \neq 0$ (such an agent always exists by Lemma \ref{lemma:zero}). Then $\agent'=(\setpStates, \causalstate_{init} = \causalstate_0,\fnTransition')$ \textit{also} solves \taskMCC, where
\begin{align}
    \fnTransition'^{(y|x)}(s)=\begin{cases}
    \fnTransition^{(y =0|0)} (s) & x = 0\\
    \fnTransition^{(y= 0|1)} (s) & x = 1\\
    \fnTransition^{(y|\Omega)}(s) & x = \Omega
    \end{cases}
\end{align}
i.e. $\fnTransition'{(y  |x)}(s)$ is the restriction of $\fnTransition{(y|x)}(s)$ to $x=\lbrace 0, 1, \Omega \rbrace$,
and measurement on receiving $\Omega$ proceeds as before;
moreover, $\agent'$ has memory cost no worse than that of $\agent$.
\end{lemma}

Since a valid input word for \taskMCC\, is also a valid input word for \taskMSC, this follows immediately.

It is also fairly simple to prove the `converse':
\begin{lemma}
\label{lemma:tally1}
Suppose a classical or quantum agent $\mathcal{A}=(\mathcal{S},s_0,\mathcal{T})$ solves \taskMCC;
assume also that $ \fnTransition^{(y|x)}=0$ for all cases where: $x \neq \Omega$, and $y \neq 0$. (Such an agent always exists, by Lemma \ref{lemma:zero}.)
Then, the agent $\mathcal{A}'=(\mathcal{S},s_0,\mathcal{T}')$ solves \taskMSC, where
\begin{align}
     \fnTransition'^{(y |x)}(s)=\begin{cases}
     \fnTransition^{(y=0|0)} (s) & x = 0\\
    \left( \fnTransition^{(y = 0|1)}\right)^x (s) & x\neq 0, \Omega\\
     \fnTransition^{(y|\Omega)}(s) & x = \Omega,
    \end{cases}
\end{align}
i.e. $T'^{(y|x)}(s)$ is the map $T^{(y|1)}(s)$ applied $x$ times,
and the measurement on receiving $\Omega$ proceeds as before;
moreover, $\mathcal{A}'$ has memory cost no worse than that of $\mathcal{A}$.
\end{lemma}

In order to convert an input timeseries for \taskMSC \, into an input timeseries for \taskMCC, we simply convert every input character $x_t \neq 0,\Omega$ into $x_t$ instances of $1$:

\begin{equation}
\strInput=243\ldots \rightarrow
\underbrace{11}_{2} \underbrace{1111}_{4} \underbrace{111}_3 \ldots
    \label{eq:repetition}
\end{equation}
Once again, the proof immediately follows.
Thus, in order to find an optimal agent solving \taskMSC, it suffices to consider \taskMCC \, instead; we can then use lemma \ref{lemma:tally1} to expand the agent into a solution for \taskMSC.

Now, since we only need consider \taskMCC, $\mathcal{T}$ can be completely specified by its behaviour on $\ket{0}\bra{0} \otimes s$, $\ket{1}\bra{1}\otimes s$ and $\ket{\Omega}\bra{\Omega}\otimes s.$ By Lemma \ref{lemma:zero}, we only need to consider agents that output $0$ on receiving any non-$\Omega$ input, and emit a class label and reset upon receiving $\Omega$.
So, for the sake of brevity, let  $\fnTransition^{(0|0)}(s) = s$ such  that inputting nothing does nothing. Meanwhile we let
\begin{eqnarray}
    \label{eq:inducedT}
    &&\fnTransOne(s)=\fnTransition^{(0|1)}(s),
\end{eqnarray}
and let $\fnMeasure$ be the measurement applied on receiving $\Omega$ input. (By Lemma \ref{lemma:projective}, we can assume $\fnMeasure$ is projective.) Then $\mathcal{T}$ is completely specified by $(\fnTransOne,\fnMeasure)$.
Thus, for brevity, in what follows we will identify the map $\mathcal{T}$ as
\begin{equation}
    \mathcal{T}=(\fnTransOne,\mathcal{M}),
\end{equation}
and identify the agent as
\begin{equation}
\label{eq:agenttuple}
    \agent=(\setpStates,\causalstate_{init} =\causalstate_0,(\fnTransOne,\mathcal{M})).
\end{equation}
(Strictly speaking, we could determine $\mathcal{S}$ from the other parameters of an \textit{optimal} agent; however, it is useful to refer to the properties of $\mathcal{S}$ during the next few results.)

\section{Optimal Classical Agent}
\label{supp:sec:classical}
We now present our first theorem, demonstrating the construction of the optimal classical agent.
First, we note that if there is only \textit{one} class, $\setSetMods = \{\setMods_1\}$ the entire problem becomes trivial, so below  we implicitly assume there is  at least 2 classes. Throughout the nontrivial case, we assume that at least two of the
sets $\setMods_j$ are nonempty. Empty classes may be removed without changing
the classification task, since no input string can have a modular sum in
an empty class. Thus, when we write $|\Xi|>1$, we mean that there are at
least two nonempty promised classes. Now an `obvious' classical construction exists: we can simply label the computational basis states $\sigma_0,\sigma_1,\dots,\sigma_{p-1}$, initialize in state $s_{init}) = \sigma_0$, and have $\fnTransOne (\sigma_j)=\sigma_{(j+1)\,\,\rm{ mod}\, \, p}$.   We will show this solution is classically optimal.

By previous results, we have reduced the problem \taskMSC \, to \taskMCC, shown that when $\charInput_t \neq \Omega$ the output can be set to $ \charOutput_t = 0$, and found that the final measurement can be assumed to be projective. In doing so, we have reduced the number of stochastic matrices we need to consider to one, namely $\fnTransOne$. So, we can use the tools of Markov process theory (see e.g. \cite{castaneda2012}) to analyze $\fnTransOne$, and prove the optimality of our construction. We begin with the following definition:
\begin{definition}
    A state $j$ is \textit{accessible} from state $i$ if there is a nonzero probability of reaching it at some time $n$: that is, $\fnTransOne^n_{ij}> 0.$ We write this relation as $i \rightarrow j$.
\end{definition}
It is easy to show that this relation is transitive: $i \rightarrow j, j \rightarrow k$ implies $i \rightarrow k$.
\begin{definition}
    States $i$ and $j$ \textit{communicate} if $i \rightarrow j$ and $j \rightarrow i$. We write this relation as $i \leftrightarrow j$.
\end{definition}
It is then fairly simple to show that this is an \textit{equivalence} relation: specifically, $\sigma \leftrightarrow \sigma'$ \textit{iff} $\sigma' \leftrightarrow \sigma,$ and $\sigma \leftrightarrow \sigma' \leftrightarrow \sigma''$ implies $\sigma \leftrightarrow \sigma''.$

Now, we can use this equivalence relation $\leftrightarrow$ to partition the set of states into \textit{communicating classes}: that is, subsets $\mathcal{C} \subseteq \{\sigma_i\}$ of the set of states such that each state in $\mathcal{C}$ communicates with each other state. In other words, if two states communicate with each other, they belong to the same communicating class. There is thus a weak ordering of communicating classes with respect to the   \textit{accessible}  relation $\rightarrow$, as two distinct communicating classes cannot, by definition, communicate with each other; in particular, if a machine leaves a communicating class, it cannot return.

We define a \textit{recurrent communicating class} as a communicating class $\mathcal{C}$ such that, when initialized in some state $\sigma \in \mathcal{C},$ the probability that we eventually return to $\sigma$ when integrated over a sufficiently long window of times where we could potentially return is one (although we do not put an upper bound on the time that this takes). It is easy to show that this must then be true for \textit{every} state in that recurrent communicating class. Furthermore, a recurrent communicating class is \textit{closed}: if we initialize in a recurrent communicating class, we will never leave. It is also simple to prove a converse statement: a finite closed communicating class is always recurrent, and in practice this is how we identify recurrent classes.

The opposite of a recurrent communicating class is a \textit{transient} communicating class: one where some state $\sigma$ (and thus every state) has a nonzero probability of \textit{never} returning. Every communicating class is either recurrent or transient. In the finite case, transient communicating classes are not closed: intuitively, there is a finite probability of escape, which means that escape is eventually certain, and returning is impossible. Figure \ref{fig:communicating} demonstrates both of these concepts.

\begin{figure}[hbtp]
    \centering
    \includegraphics[width=0.9\columnwidth]{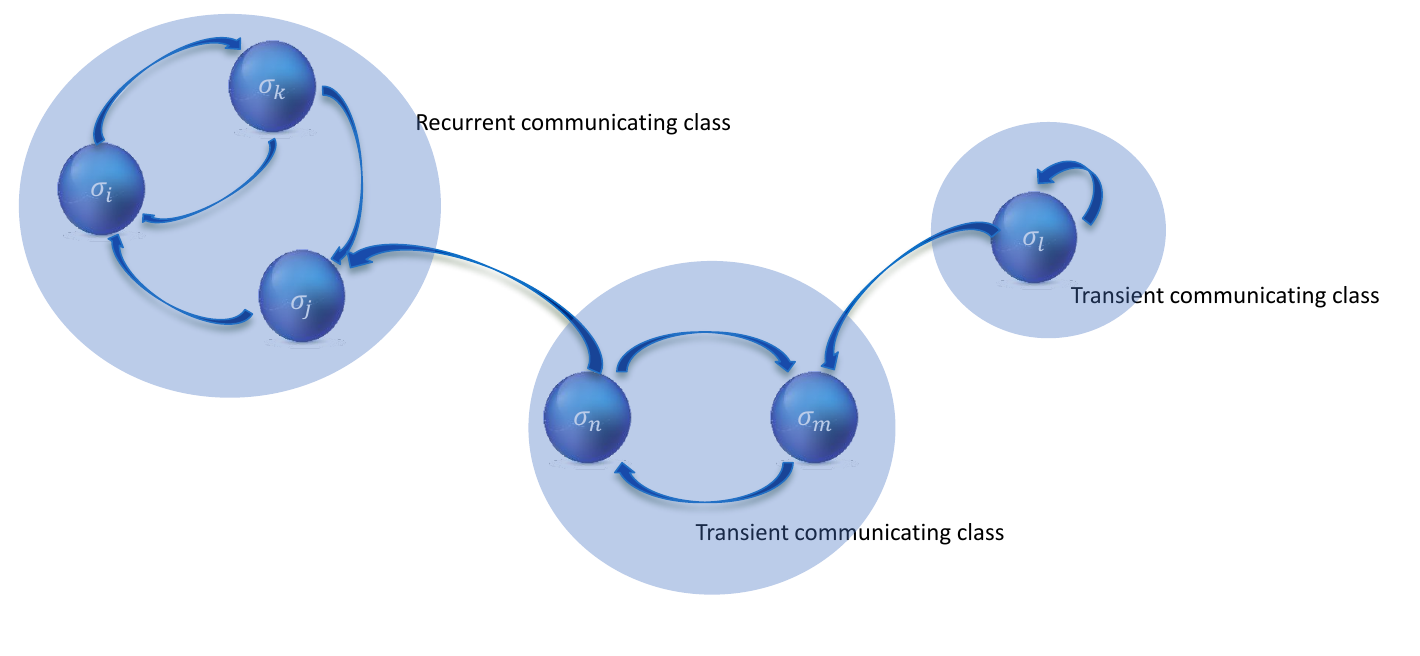}
    \caption{An example Markov process with communicating classes highlighted as blue regions -- here a communicating class corresponds to a collection of states $\mathcal{C} = \{\sigma_i\}$ such that, for any $\sigma_i, \sigma_j \in \mathcal{C}$, there exists $n,m \in \mathbb{Z}^+$ such that  $G_{ij}^n >0$ and $G_{ji}^m >0 $ (i.e. $\sigma_i$ is accessible from $\sigma_j$ and vice versa). The recurrent class is a closed communicating class, while the transient classes are not. By closed we mean that, for any $\sigma_i \in \mathcal{C}$ where $\mathcal{C}$ is a recurrent communicating class, if  $\sigma_j$ is accessible from $\sigma_i$, then $\sigma_j \in \mathcal{C}$.}
    \label{fig:communicating}
\end{figure}

Lastly, a Markov process is \textit{irreducible} if it consists of a single communicating class. For Markov processes with a finite number of states, this communicating class must be recurrent. This definition is equivalent to irreducibility of stochastic matrices; if a Markov process is irreducible, the corresponding matrix is irreducible, and vice versa.

With these definitions, we are now ready to proceed with the main result of this section. Let us identify agents using the more specific tuple as in Eq. \eqref{eq:agenttuple}. We then wish to prove a couple key properties of $\fnTransOne$:
\begin{lemma}
\label{lemma:classicalirreducible}
Suppose an \textit{optimal} classical agent $\agent = \tupleAgentShort$ solves \taskMCC. Then there exists an optimal classical agent $\agent' = (\setpStates',\causalstate'_{init} = \causalstate'_0,(\fnTransOne',\fnMeasure'))$ such that $\fnTransOne'$ is irreducible.
\end{lemma}

\begin{proof}
If $G$ is irreducible, the lemma is trivial. Let us then suppose that $\mathcal{A}$ is not irreducible. That is, assume the transition matrix $G$, acting on the computational basis states $\sigma_i \in S_b$, generates a stochastic process that has two or more communicating classes.

Since prepending the input string with $p$ $1$'s should not affect the final answer (which is computed mod $p$), the agent $(\setpStates, \rho_0 = G^{np}(s_0), (G, \mathcal{M}))$ will also solve the problem for any integer $n$.

The memory state $\rho_0$ is a classical mixed state, which can be expressed as a convex combination of basis states $\rho_0 = \sum p_i \sigma_i$. For sufficiently large $n$, the probability mass of the Markov process must eventually concentrate in the recurrent states. Thus, there is a strictly positive probability $p_j > 0$ associated with some basis state $\sigma_j$ that belongs to a recurrent communicating class $\mathcal{C}$.

By invoking the purity lemma (Lemma \ref{lemma:purity}), we can simply pick this pure state $s'_0 = \sigma_j$ in the support of $\rho_0$. Because $s'_0 \in \mathcal{C}$, and communicating classes are disjoint, the state lies entirely within a single recurrent class. Furthermore, because recurrent communicating classes are closed under the action of $G$, the agent will never transition to a basis state outside of $\mathcal{C}$.

Consider the new agent $\mathcal{A}'$ where $s_0$ is replaced by $s'_0 \in \mathcal{C}$; then $\Hacc(\agent') \subseteq \rm{span}(\mathcal{C}) \subseteq \Hacc(\agent),$
with the memory space $\Hacc(\agent').$
The resulting construction $\agent'$ has memory cost $\dimHacc(\agent')$ no greater than that of the optimal agent $\mathcal{A}$. Since it only contains one communicating class, its transition matrix $G'$ is irreducible, completing the proof.
\end{proof}
In short, while a given optimal agent may not have an irreducible $\fnTransOne,$ we may assume that an optimal agent with irreducible $\fnTransOne$ exists. We will make use of Lemma \ref{lemma:classicalirreducible} to drastically simplify the following proofs.

We now make use of another concept from Markov process theory, the \textit{period}, which will help us understand which $\fnTransOne$ are possible. By Lemma \ref{lemma:purity}, we know that an optimal classical agent with a pure initial state must exist; we now focus on the period of the initial state.
\begin{lemma}
\label{lemma:classicalperiod}
Suppose an \textit{optimal} classical agent $\agent = \tupleAgentShort$ with pure initial state $s_0 = \sigma_0$ and irreducible $\fnTransOne$ solves \taskMCC. Define the \textit{period} $h$ of the stochastic matrix $\fnTransOne$ on state $\sigma_0$ as the greatest common divisor of all $n$ such that $\fnTransOne^n(\sigma_0)$ has nonzero probability of being in $\sigma_0$.  Then $h \ge p$.
\end{lemma}
\begin{proof}

Let $\mathcal{N} = \{n_i\}_{i=1}^\infty$ be the set of all integers $n$ such that $\fnTransOne^n(s_0)$ has nonzero probability of being in $\sigma_0$. By considering the composition $\fnTransOne^{n+n'}(\sigma_0)=(\fnTransOne^n \circ \fnTransOne^{n'})(\sigma_0),$ we observe that $\mathcal{N}$ is closed under addition.

Next, let $g(j) = \textrm{gcd}(n_1, \dots, n_j);$ then $g(j)$ is a monotonically decreasing function to the positive integers, which implies that it achieves its minimum $h$ at some finite $\tilde{j}$. That is, there is some $\tilde{j}$ such that $h=\textrm{gcd}(n_1, \dots, n_{\tilde {j}})$. We have therefore found a finite subset of $\mathcal{N}$ that determines its greatest common divisor $h$.

Now, since $h$ is the greatest common divisor of $\{n_1, \dots, n_{\tilde {j}}\}$, we can divide each element by $h$ to yield a set of setwise coprime integers $\{n_1/h, \dots, n_{\tilde {j}}/h\},$ i.e. $\gcd \{n_1/h, \dots, n_{\tilde {j}}/h\}=1.$
Schur's theorem states that, for \textit{all} sufficiently large integers $k > k_{crit}$, there exist non-negative integers $c_i$ such that $\sum_{i=1}^{\tilde{j}} c_i n_i/h=k$ \cite{alfonsin2005} \footnote{While \textit{finding} $k_{crit}$ is known as the Frobenius problem \cite{alfonsin2005}, and can be computationally difficult, we only need to know that it exists; various upper bounds have been found by e.g. \cite{gasarch2015}.}. Recalling that $\mathcal{N}$ is closed under addition, it therefore follows that $hk \in \mathcal{N}$ for all $k > k_{crit}.$

Suppose, for contradiction, that $p>  h$. Since $p$ is prime,
$\gcd(h,p)=1$. Consider some distinct $\setMods_1$ and $\setMods_2$ in $\setSetMods$.
Let $\xi_1 \in \setMods_1$ and $\xi_2 \in \setMods_2$. Note also that $\intPeriod$ and $p$ are co-prime (since $p$ is prime). The Chinese Remainder Theorem states that, given a collection of pairwise coprime positive integers $\{q_i\}_i$ and remainders $0\le a_i < q_i,$ there exists exactly one solution $0 \le x < \prod_i q_i$ such that   $(x \mod q_i) = a_i$ for all $i$. Consider the coprime integers $\{h,p\}$. The Chinese Remainder Theorem therefore implies there exists a unique solution $\chi_1$ to the following system of modular equations:
\begin{align}
\label{eqn:crt1}
    \chi_1 \mod p &= \xi_1\nonumber\\
    \chi_1 \mod h &= 0;
\end{align}
We then apply the theorem again to find a unique solution to the following system:
\begin{align}
\label{eqn:crt2}
    \chi_2 \mod p &= \xi_2\nonumber\\
    \chi_2 \mod h &= 0.
\end{align}
In particular, since $h$ divides $\chi_1,\chi_2$, we may write $\chi_1=k_1 h, \chi_2 = k_2 h$ for integer $k_1, k_2.$
These equations are also satisfied if we add on any integer multiple of $hp$; specifically, choose
\begin{align}
    \chi'_1&=\chi_1+hp(k_{crit}+1) = h k_1' \\
    \chi'_2&=\chi_2+hp(k_{crit}+1) = h k_2'
\end{align}
where  $k'_1=k_1+p(k_{crit}+1), k'_2=k_2+p(k_{crit}+1).$
Since $k'_1,k'_2 > k_{crit}$, it follows from Schur's theorem that $\chi'_1,\chi'_2 \in \mathcal{N}$.

Thus, upon input $\Omega$, after receiving a sequence of $\chi'_1$ ones, or of $\chi'_2$ ones, the agent's memory will return to the exact same computational basis state $\sigma_0$ (some basis state inside $s_0$) with strictly positive probability.

Because the final measurement $\mathcal{M}$ depends only on the current state, arriving at the identical basis configuration $\sigma_{0}$ guarantees an identical conditional probability distribution over the output labels.

However,  $\chi_1' \textrm{  mod  } p  \equiv \xi_1 \in \Xi_1$ and $ \chi_2' \textrm{  mod  } p \equiv \xi_2 \in \Xi_2$; thus any agent solving $\taskMCC$ is required to deterministically discriminate between them. This contradicts the previous paragraph. As such, $h$ cannot be co-prime with $p$. This contradicts the assumption $p>h$; in particular, $h \ge p$.

\end{proof}

As a side note, by Lemma \ref{lemma:classicalirreducible}, we may assume $\fnTransOne$ is irreducible, and so \textit{every} pure state has the same period; we may therefore simply say that $h$ is \textit{the} period of $\fnTransOne$  \cite{pishronik2014}.

We can go even further. Consider \textit{any} classical agent $\mathcal{A}$ with memory cost less than $p$. We now demonstrate that there exists input words $\strInput^{(1)}, \strInput^{(2)}$ with $\fnClass{\strInput^{(1)}}\neq \fnClass{\strInput^{(2)}}$ such that the resulting memory states of $\mathcal{A}$ are arbitrarily close.

\begin{lemma}
\label{lemma:smallclassical}
Suppose a classical agent $\agent = \tupleAgentShort$ has memory cost less than $p$. Then for any choice of $\fnTransOne$ and any two distinct classes $\setMods_1, \setMods_2$, we can always find two input strings  $\strInput^{(1)}, \strInput^{(2)} $, with $\fnClass{\strInput^{(1)}}\neq \fnClass{\strInput^{(2)}}$ , and corresponding sums $n_1, n_2$ such that $\fnTransOne^{n_1}(s_0), \fnTransOne^{n_2}(s_0)$ are arbitrarily close.
\end{lemma}

\begin{proof}
First consider the case where $\fnTransOne$ is irreducible, and has period $h$.
The Perron-Frobenius theorem for irreducible matrices is as follows (\cite{boylenotes} Thm. 5.9, paraphrasing slightly): Let $A$ be an irreducible matrix of period $h$. Then:
\begin{itemize}
    \item $A$ has a nonzero right eigenvector $r$. This eigenvector is strictly positive (i.e., every entry is positive).
    \item Its eigenvalue $\lambda$ is the spectral radius of $A$, and any nonnegative eigenvector of $A$ is a scalar multiple of $r$.
    \item The characteristic polynomial of $A$ has exactly $h$ roots on its spectral circle (the peripheral spectrum). These roots are simple, and they are precisely $\lambda, \lambda e^{2\pi i/h}, \lambda e^{4\pi i/h}, \dots, \lambda e^{2(h-1)\pi i/h}$.
\end{itemize}

In particular, the characteristic polynomial of $\fnTransOne$ has $h$ roots, and thus $\fnTransOne$ must have at least $h$ distinct eigenvalues, and so $h \le \dimHacc < p$.

We can go further. By \cite{boylenotes}, every irreducible matrix $A$ can be put into \textit{cyclic block form} via some permutation matrix $P$. For instance, if $\fnTransOne$ has period 3, then there exists some permutation $P$ such that
\begin{equation}
    PA P^{-1} = \begin{bmatrix}
        0 &\fnTransOne_1 &0 \\
        0 &0&\fnTransOne_2 \\
        \fnTransOne_3 &0 &0
    \end{bmatrix},
\end{equation}
for some (not necessarily square) blocks $\fnTransOne_1, \fnTransOne_2, \fnTransOne_3$. We call this the \textit{cyclic block form} of $\fnTransOne$.

Prop.5.5 of \cite{boylenotes} states the following: Let $A$ be an irreducible matrix of period $h$ in cyclic block form.
Then $A^h$
is a block diagonal matrix and each of its diagonal blocks is a primitive matrix, i.e. irreducible and period 1.
Moreover, the diagonal blocks have the same nonzero spectrum.

Applying this, we see that there exists some permutation matrix $P$ such that $P\fnTransOne^h P^{-1}$ is block diagonal, in the form
\begin{equation}
    P\fnTransOne^h P^{-1} = \begin{bmatrix}
        \fnTransOne_1 &0 &0 &\dots\\
        0 &\fnTransOne_2 &0 &\dots\\
        0 &0 &\fnTransOne_3 &\dots\\
        \vdots &\vdots &\vdots &\ddots
    \end{bmatrix},
\end{equation}
where each block matrix $\fnTransOne_l$ is primitive, i.e. irreducible and period 1. In particular, as $n \to \infty$, the powers of these primitive matrices inside $G^{nh}$ approach rank-one projections onto their respective stationary distributions;  thus $\fnTransOne^{nh}$  approaches a block-diagonal projection $\Pi_h$ onto the $h$-step recurrent blocks.

We then proceed as in Lemma \ref{lemma:classicalperiod}; as $h < p,$ $h$ and $p$ are coprime. Therefore, given $\xi_1 \in \setMods_1, \xi_2 \in \setMods_2$, we can once again find sufficiently large solutions $n_1, n_2$ to \eqref{eqn:crt1}, \eqref{eqn:crt2} such that $\fnTransOne^{n_1}(s_0)$ and $\fnTransOne^{n_2}(s_0) $ are arbitrarily close: specifically, we can pick $n_1$ and $n_2$ such that $\fnTransOne^{n_1} \approx \Pi_h$ \textit{and} $\fnTransOne^{n_2} \approx \Pi_h$. Then we simply set $\strInput^{(1)}$ to be $n_1$ 1's followed by $\Omega$, and $\strInput^{(2)}$ to be $n_2$ 1's followed by $\Omega$.

We did not assume agent $\agent$ is optimal, nor that it `solves' $\taskMCC$ with zero error. We must therefore consider the more general case where $\fnTransOne$ is reducible. We may again assume $\mathcal{H}=\Hacc(\agent)$; so $\dim\mathcal{H}<p$.

For $n$ sufficiently large, the probability of $\fnTransOne^n (s_0)$ remaining in a transient state approaches zero: by definition, the expected number of times one visits a transient state is finite \cite{castaneda2012}, and there are only finitely many such states. Specifically, for any $\eta > 0$, there exists $N$ such that $\fnTransOne^n(s_0)$ has probability less than $\eta$ of being in a transient state for all $n > N$.

Let $\Pi_R$ be the projector onto the recurrent states, and let $v_N = \Pi_R\fnTransOne^N(s_0).$ Then for $n > N$,
\begin{equation}
    \fnTransOne^n(s_0) = G^{n-N}v_N + G^{n-N}(1-\Pi_R)G^N(s_0),
\end{equation}
where the second term vanishes as $\eta \rightarrow 0, N \rightarrow \infty$; thus $v_N \approx \fnTransOne^N(s_0)$, and so for $n > N$, $\fnTransOne^n(s_0) \approx \fnTransOne^{n-N}(v_N)$

Now, consider the \textit{normal form of the reducible matrix} $G$; that is, for some permutation matrix $P$, there exists the block-triangular form
\begin{equation}
    PG P^{-1} \rightarrow \begin{bmatrix}
        G_1 &* &* &\dots\\
        0 &G_2 &* &\dots\\
        0 &0 &G_3 &\dots\\
        \vdots &\vdots &\vdots &\ddots
    \end{bmatrix},
\end{equation}
where each \textit{diagonal} block is either irreducible or simply a 1$\times$1 zero matrix; this form can be found by repeatedly `breaking down' reducible diagonal blocks into smaller blocks (\cite{boylenotes} Thm. 6.1).

We can constrain the normal form of this reducible matrix further.
Suppose some block, say $G_2$, has a nonzero value above it in its column.
We observe that by the first part of the Perron-Frobenius theorem, $G_2$ has a strictly positive eigenvector with eigenvalue $\lambda$ equal to the spectral radius. Obviously, that eigenvalue cannot be greater than one, since $G$ describes a stochastic process. That eigenvalue also cannot be equal to one, because the nonzero value above the block implies a nonzero probability of `escaping' the block. So if a block has a nonzero value above, that block is transient; and the converse is also true.

Let us now consider the action of $P\fnTransOne P^{-1}$ only on the recurrent states.
By the previous argument, $P\fnTransOne P^{-1}$ acting on the recurrent states (i.e. the restriction of the matrix onto those states) is block-diagonal:
\begin{equation}
    P\fnTransOne P^{-1} = \begin{bmatrix}
        D_1 &0 &0 &\dots\\
        0 &D_2 &0 &\dots\\
        0 &0 &D_3 &\dots\\
        \vdots &\vdots &\vdots &\ddots
    \end{bmatrix},
\end{equation}
where each diagonal block $D_l$ is irreducible. By previous arguments, each $D_l$ has some period $h_l < p;$ so the least common multiple $h$ of these periods is still coprime to $p$. In particular, the corresponding `blocks' of $PG^h P^{-1}$ are all aperiodic.

Let us now additionally assume $N=Lh$ for some integer $L$. (Since $N$ is large but otherwise arbitrary, this is fine.)
We then effectively repeat the previous construction over all these blocks simultaneously: as $n \rightarrow \infty$, we therefore find that the action of $P\fnTransOne^{h(n-L)} P^{-1}$  on the recurrent states approaches a projection, and so $\fnTransOne^{h(n-L)} \rightarrow \Pi_h$, at least on the recurrent states; thus $G^{hn}(s_0) \approx \Pi_h v_N.$

Given $\epsilon>0$, choose $N=Lh$ sufficiently large that
$v_N$ is within $\epsilon/4$ of $\fnTransOne^N(s_0)$. Because $\fnTransOne$ is stochastic,
for every $n>N$, $\fnTransOne^n(s_0)$ is within $\epsilon/4$ of
$G^{n-N}(v_N)$. (Specifically, because of the $l_1$ non-expansiveness of stochastic matrices, the classical analogue of the Data Processing Inequality.)

We then find  $n_1, n_2$ solving the modular equations such that $\fnTransOne^{n_1-N}(v_N)$ and $\fnTransOne^{n_2-N}(v_N)$ are both within $\epsilon/4$ of $\Pi_h(v_N)$, once again using the fact that $h$ is coprime to $p$. Combining these distances, we find that $\fnTransOne^{n_1}(s_0)$ and $\fnTransOne^{n_2}(s_0)$ are within $\epsilon$ of each other. The rest follows as in the irreducible case.
\end{proof}

Lemma \ref{lemma:smallclassical} states that given a classical $\mathcal{A} = (\mathcal{S}, s_{init} =s_0, (\fnTransOne, \mathcal{M}))$  with memory dimension $h = \dimHacc(\agent) < p$ we can always find two different input strings $\strInput^{(1)},\strInput^{(2)}$, with distinct classes $\fnClass{\strInput^{(1)}}\neq \fnClass{\strInput^{(2)}}$, where the inputs $\strInput^{(1)},\strInput^{(2)}$ consist of $n_1$ one's and $n_2$ one's respectively, such that the induced states on the agent's memory $\fnTransOne^{n_1}(s_0)$ and $\fnTransOne^{n_2}(s_0)$ are arbitrarily close. The data processing inequality implies that the distance between $\fnTransOne^{n_1}(s_0)$ and $\fnTransOne^{n_2}(s_0)$ bounds the distance on any outputs obtained by measuring these two states. As such, upon  receiving the termination character $\Omega$, $\agent$ cannot respond differently on these two inputs, and thus cannot solve $\taskMCC$. Therefore, the agent has to guess the appropriate class when it receives $\strInput^{(1)}, \strInput^{(2)}$.

Suppose that the agent's response has no correlation at all with its internal state $s$. And assume that $\setData$ is supported on  $\bigcup_{j=1}^K\Gamma_j$, so that it is possible to obtain non zero loss.
In this case, the lowest possible misclassification risk takes a simple form:
\begin{lemma}
\label{lemma:guess}
    Consider \taskTSC. Consider a distribution $\setData(\strInput)$. Then there exists an lower bound $\eRisk_G$ on the expected loss of `guessing the solution' to \taskTSC:
    that is, for any classifier agent $\agent=\tupleAgent$ such that $\fnMeasure(s)=\fnMeasure(s_0)$ for all $s \in \setpStates$, then the expected loss has lower bound
    \begin{equation}
        \eRisk_G \equiv 1-\max_j P(\vecOutput=j | \strInput \sim \setData(\strInput)).
    \end{equation}

    So an agent that guesses will incur at least this much loss.
\end{lemma}
\begin{proof}First before guessing the agent $\mathcal{A}$ will have a (Bayesian) prior $\aleph(\strInput)$ over the strings. The best-case scenario is when $\aleph(\strInput) = \setData(\strInput)$. The result then follows as a simple consequence of the fundamental theorem of linear programming: that is, since the expected risk is an affine function of the output probabilities, the optimum is achieved on an extremum, which in this case is where the agent simply picks the same thing every time. The best guess is $ \hat{C}(\strInput) =  \textrm{argmax}_j P(j | \strInput \sim \setData(\strInput))$.
\end{proof}
In particular, if there are only two possible classifications, with equal probability, this lower bound is simply 1/2.

Finally, we now prove the optimality of our simple classical construction, in an extremely strong way: any classical agent with less memory does little better than guessing, with expected loss bounded away from zero.
\begin{theorem}
\label{thm:optimalclassical}
Suppose an \textit{optimal} classical agent $\mathcal{A}=\tupleAgentShort$ solves \taskMCC, for $\Xi$ nontrivial, i.e. $|\Xi| > 1$. Then the memory cost of $\agent$ is $p.$
Furthermore, for any classical agent $\mathcal{A}'=(\setpStates',\causalstate'_{init} = \causalstate'_0,(\fnTransOne',\fnMeasure'))$ with memory cost less than $p$, and $\epsilon > 0$, there exists a distribution $\setData(\strInput)$ such that the expected loss $\eRisk_{\agent'}$ on $\taskMSC$ is at least
\begin{equation}
    \eRisk_G - \epsilon = \frac{1}{2}-\epsilon.
\end{equation}
\end{theorem}
\begin{proof}

Let $\agent$ be a memory optimal classical agent which solves \taskMCC. By Lemma \ref{lemma:classicalirreducible}, its stochastic matrix $\fnTransOne$ is irreducible. Suppose $\fnTransOne$ has period $h$: by the Perron-Frobenius theorem, $\fnTransOne$ has $h$ distinct eigenvalues $e^{2\pi\ii j/h}$  \cite{seneta2008}. In particular, the matrix must have dimension at least $h$. By lemma \ref{lemma:classicalperiod}, $h \ge p$. Since we already constructed an agent with memory cost $p$, we have thus demonstrated that the optimal classical agent has $h=p$, and memory cost $p$.

For classical agents $\mathcal{A}'$ with memory dimension less than $p$, for simplicity, we consider $\taskMCC$. Lemma \ref{lemma:smallclassical} implies the existence of $\strInput^{(1)}, \strInput^{(2)}$, such that
\begin{itemize}
    \item $\fnClass{\strInput^{(1)}} \neq \fnClass{\strInput^{(2)}}$
    \item $\sum_t \charInput^{(1)}_t = n_1, \sum_t \charInput_t^{(2)} = n_2$, and
    \item $\fnTransOne'^{n_1}(s_0), \fnTransOne'^{n_2}(s_0)$ are arbitrarily close.
\end{itemize}
More precisely, for the last point, let us pick $\strInput^{(1)}, \strInput^{(2)}$ such that the trace distance
\begin{equation}
    \mathrm{D}(\fnTransOne'^{n_1}(s_0), \fnTransOne'^{n_2}(s_0)) = \frac{1}{2}\|\fnTransOne'^{n_1}(s_0)-\fnTransOne'^{n_2}(s_0)\|_1 <\eta
\end{equation}
 where $\eta > 0$ is a constant which will be determined later. Without loss of generality, let us assume $\fnClass{\strInput^{(1)}}=1, \fnClass{\strInput^{(2)}} = 2$. Since $\fnTransOne'^{n_1}(s_0)$ is a stochastic vector
\begin{equation}
    \mathrm{D}(\fnTransOne'^{n_1}(s_0), \fnTransOne'^{n_2}(s_0)) = \frac{1}{2}\sum_j |\fnTransOne'^{n_1}(s_0)_j-\fnTransOne'^{n_2}(s_0)_j|.
\end{equation}

Now, consider the possible actions that the agent may execute using its internal state. It may apply any measurement, followed by some classical post-processing, and then output the result.  For clarity, let us write $\mathbf{P}^{(1)}$ to represent the probability vector corresponding to $\fnGuess{\strInput^{(1)}}$ (i.e. $P^{(1)}_j=Pr[\fnGuess{\strInput^{(1)}}=j]$), and similarly $\mathbf{P}^{(2)}$ for $\fnGuess{\strInput^{(2)}}$.
Writing the composition of the measurement and post-processing as a channel $\mathcal{E}$,  the data processing inequality implies the trace distance
\begin{align}
    \mathrm{D}(\mathbf{P}^{(1)}, \mathbf{P}^{(2)})&=\mathrm{D}(\mathcal{E} \circ \fnTransOne'^{n_1}(s_0),\mathcal{E}\circ \fnTransOne'^{n_2}(s_0))\nonumber\\
     & \le  \mathrm{D}(\fnTransOne'^{n_1}(s_0), \fnTransOne'^{n_2}(s_0))\nonumber\\
    &< \eta.
\end{align}
Since $\mathrm{D}(\mathbf{P}^{(1)},\mathbf{P}^{(2)}) < \eta,$ we find
\begin{align}
    \frac{1}{2}\sum_j |P^{(1)}_j-P^{(2)}_j|<\eta;
\end{align}
Furthermore, since we are optimizing over simple misclassification, the optimal guess will only include the possible classifications $j=1,2,$ and so we find that $P_2^{(k)}=1-P_1^{(k)}$, and thus
\begin{align}
    |P^{(1)}_1-P^{(2)}_1|&=|P^{(1)}_2-P^{(2)}_2|,\\
    |P^{(1)}_1-P^{(2)}_1| &< \eta.
\end{align}

Now, let us simply choose $\setData$ such that $\strInput^{(1)},\strInput^{(2)}$ occur with probability 1/2 each.
Using this dataset, the probability of misclassification of $\agent'$ on $\taskMCC$ is
\begin{align}
    \eRisk_{\agent'}&=\frac{1}{2}\left(P^{(1)}_2+P^{(2)}_1\right)
\end{align}
Since this expression is continuous (linear, even), for all $\epsilon>0$ there exists some $\eta>0$ such that whenever $\mathrm{D}(\mathbf{P}^{(1)},\mathbf{P}^{(2)}) < \eta,$ the expected loss is within $\epsilon$ of the expected loss when $\mathbf{P}^{(1)}=\mathbf{P}^{(2)},$ which is simply 1/2. That is, $\mathrm{D}(\mathbf{P}^{(1)},\mathbf{P}^{(2)}) < \eta \implies\left|\fnELoss{\agent'}-\frac{1}{2}\right| < \epsilon.$
More precisely,
\begin{align}
    \left|\fnELoss{\agent'}-\frac{1}{2}\right|
    &=\left|\frac{1}{2}\left(P^{(1)}_2+P^{(2)}_1\right) - \frac{1}{2}\left(P^{(1)}_2+P^{(1)}_1\right)\right|\\
    &=\left|\frac{1}{2}(P^{(2)}_1-P^{(1)}_1)\right|\\
    &<\frac{1}{2}\eta.
\end{align}
Thus, to ensure that this is less than $\epsilon$, $\eta = 2\epsilon$ suffices.

Finally, Lemma \ref{lemma:guess} shows that for all input-independent strategies, i.e. $\mathbf{P}^{(1)}=\mathbf{P}^{(2)}$, the expected loss is at least $\eRisk_G = 1/2$, and so $\fnELoss{\agent'} > 1/2 - \epsilon,$ completing the proof. (Strictly speaking, Lemma \ref{lemma:guess} holds for $\taskMSC$ instances, but these instances are also $\taskMCC$ instances, by Lemma \ref{lemma:tally2}.)

\end{proof}

While Lemma \ref{lemma:smallclassical} designs a distribution adversarially, i.e. against a particular $\agent$,
we note that for any $\setSetMods,$ it is also possible to construct a natural family of input distributions $\{\setData_\alpha\}$ such that, for \textit{any} classical agent with memory dimension $M < p$, for sufficiently small $\alpha$, the agent does no better than guessing. In this case, $\mathcal{L}_G = 1-1/K$  (a proof of this claim is included in  Sec. \ref{sec:naturalguess}). The main idea is to add a long random pad of $pN$ ones, selecting the pad length from a geometric distribution over $N$ with stopping probability $\alpha$; the pad must be long enough that $G^{pN}(s_0)$ is almost certainly supported on recurrent classes, and long enough to be evenly distributed within each class, for \textit{any} input. For further information, see Sec. \ref{sec:naturalguess}.

\section{Constructing a Quantum Agent}
\label{sec:construction}
We next demonstrate a construction of a quantum agent with memory cost less than $p$. This suffices to demonstrate quantum advantage. As we will demonstrate, this cost is bounded above by $(\sum_j|\setMods_j|)^2$; in particular, the quantum advantage enjoyed by these quantum agents can be arbitrarily large, for certain $\setSetMods$ and sufficiently large $p$. (This bound is by no means tight, but it is the simplest to state.)

Consider the construction outlined in Algorithm \ref{alg:construction}. Firstly note that there are superficial differences between the algorithm in the main body of the paper and Algorithm \ref{alg:construction}: in particular, instead of $\mathcal{N}_{\setSetMods}$ we use $\Delta$, and instead of $\ket{\sigma_j}$ we use $\ket{s_j}$. These differences do not change the output, they simply make it notationally easier to prove the algorithm's correctness. Secondly looking at Algorithm \ref{alg:construction}, we note that since this construction has the map updating the memory state, $\fnTransOne(\rho)$, be \textit{unitary} (i.e. Kraus rank 1), we will write $U\rho U^\dagger=\fnTransOne(\rho).$
We can then construct the quantum agent $\mathcal{A}=(\mathcal{S},\ket{s_0},(\fnTransOne,\mathcal{M}))$ solving \taskMCC \, using Algorithm \ref{alg:construction}.

\begin{algorithm}
\caption{Quantum Agent Construction} \label{alg:construction}
\SetKwInOut{Input}{input}
\SetKwInOut{Output}{output}
\SetInd{0.5em}{1.5em}
\SetAlgoHangIndent{1.5em}
\Input{$p,\Xi=\lbrace \Xi_j\rbrace_{j=1}^K$}
\Output{$\mathcal{A} =  (\mathcal{S}, s_{init} ,\mathcal{T}=(\fnTransOne  ,\mathcal{M}))$}
\Begin{
    define $\delta_0=0$\;
    construct for all classes $\setMods_i, \setMods_j$ the difference sets
    \begin{align}
    \setDiffs_{ij}=\lbrace \delta=\xi_1-\xi_2 \bmod p  \,&|\,\xi_1\in\setMods_i, \xi_2\in\setMods_j,\nonumber\\
      &0<\delta<p/2\rbrace,
    \end{align} between $\setMods_i, \setMods_j$\;
    construct the difference set
    \begin{align}
     \setDiffs = \bigcup_{i \neq j}\setDiffs_{ij}\cap \left\{0,1,\ldots,\frac{p-1}{2}\right\}.
    \end{align}
    and enumerate $\lbrace \delta_1, \delta_2,\dots,\delta_{|\Delta|}\rbrace=\Delta$\;
    let $\theta = 2\pi/p$\;
    define ($(2|\Delta|+1) \times p$) matrix $A$ such that \begin{equation}\label{eq:Amatrix}
    A_{kj}=\begin{cases}
        \cos j\delta_{k'}\theta & k=2k' \textrm{ for } k'\in \mathbb{Z},\\
        \sin j\delta_{k'}\theta & k=2k'-1 \textrm{ for } k'\in \mathbb{Z}
        \end{cases}
    \end{equation}
    for $k=0,1,\dots,2|\Delta|$\;
    let $\mathbf{e}_0 \in \mathbb{R}^{2|\Delta|+1}$ be a $(2|\Delta|+1)$-dimensional vector with $1$ on its first entry and $0$ elsewhere\;
    find nonnegative solution vector $\mathbf{b}\in \mathbb{R}^p$, $\mathbf{b}\ge 0$ to matrix equation
    \begin{equation}\label{eq:abequation}
    A\mathbf{b}=\mathbf{e}_0
    \end{equation}
    with at most $2|\Delta|+1$ nonzero entries, e.g. using simplex method \cite{stone1991simplex}; by Theorem \ref{thm:quantumdim} this is always possible\;
    let $M \le 2|\Delta|+1$ be the number of nonzero entries in $\mathbf{b}$\;
    let $l_i$ be the row-number of the $i$-th nonzero element of $\mathbf{b}$ for $i = 1,2,\ldots$\;
    let $\mathcal{H}$ be a Hilbert space with an orthonormal basis $\lbrace \ket{\phi_1},\ket{\phi_2},\dots,\ket{\phi_M}\rbrace$\;
    let $\mathcal{S}$ be the density matrices over $\mathcal{H}$\;
    let $\ket{s_0}=\sum_{i=1}^M \sqrt{b_{l_i}}\ket{\phi_i}$\;
    define $U$ such that $U \ket{\phi_j} = e^{\ii l_j \theta} \ket{\phi_j}$\;
    let $s_{init} =  \ket{s_0}\bra{s_0}$\;
    let  $\fnTransOne(\rho) = U\rho U^{\dagger}$\;
    let $\ket{s_j}=U^j\ket{s_0}$ for $j=0,1,2,\dots,p-1$ \;
    let $\mathcal{M}$ be the measurement according to the observable $\sum_{j=1}^K j P_j$, where $P_j$ is the projection operator onto $S_j=\mathrm{Span}_{k \in \Xi_j}\{\ket{s_k}\}$; complete the measurement basis with a rejection operator $P_0 = I - \sum_j P_j$ with output 0.
}
\end{algorithm}

To prove the correctness of Algorithm \ref{alg:construction}, we recall the definition from the main body:
\begin{definition}
The \textit{difference set between classes $\setMods_i, \setMods_j$}  is
\begin{align}
    \setDiffs_{ij}=\lbrace \delta=\xi_1-\xi_2 \bmod p  \,&|\,\xi_1\in\setMods_i, \xi_2\in\setMods_j,\nonumber\\
      &0<\delta<p/2\rbrace.
\end{align}
The \textit{difference set}
\begin{align}
     \setDiffs = \bigcup_{i \neq j}\setDiffs_{ij}\cap \left\{0,1,\ldots,\frac{p-1}{2}\right\}.
\end{align}
is the union of the difference sets between classes over all possible pairs of distinct classes, after filtering out `duplicates'.
\end{definition}

Note that in the main text we have also introduced a larger set $\mathcal{N}_{\setSetMods} = \{(\xi_l - \xi_m) \mod p| \forall  \, \xi_l \in \setMods_i, \xi_m \in \setMods_j,  \textrm{ and  }\forall \,\,\setMods_i, \setMods_j \in \setSetMods \textrm{ such that }i\neq j\}$. It is easy to check that $|\mathcal{N}_\setSetMods| = 2 |\Delta|$. This follows because $\xi_i - \xi_j \in \Delta$ implies  $[(\xi_j - \xi_i) \textrm{  mod  } p] > p/2 $, and thus  $\xi_j - \xi_i \notin \Delta$. And thus there are two contributions to $\mathcal{N}_\setSetMods$ for every element in $\Delta$. We leave $\Delta$ here as it is implicit in the construction of Algorithm  \ref{alg:construction}.

Using this definition, we can prove the following lemma:
\begin{lemma}
\label{lemma:delta}
Let $\Delta$ be the difference set for $\setSetMods$. Then
$\braket{s_0}{s_\delta}=0$ for \textit{all} $\delta \in \Delta$ iff $\braket{s_{\xi_1}}{s_{\xi_2}}=0$ for all $\xi_1,\xi_2$ in distinct classes $\Xi_i,\Xi_j \subseteq \mathbb{Z}_p$.
\end{lemma}

\begin{proof}
For any $\xi_1, \xi_2$ in distinct classes $\Xi_j, \Xi_k$, by construction, the difference set $\Delta$ contains either this difference $\xi_2-\xi_1$ or the reverse $\xi_1-\xi_2.$ Let us first consider the former case, i.e. $\Delta \ni \delta = \xi_2 - \xi_1.$ Expanding the definition of $s_\delta$ we find
\begin{align}
\braket{s_0}{s_\delta} &= \ematrix{s_0}{U^\delta}{s_0} \\
&= \ematrix{s_0}{U^{\xi_2-\xi_1}}{s_0} \\
&= \braket{s_{\xi_1}}{s_{\xi_2}}.
\end{align}

Both the forward and reverse implications follow almost immediately: if $\braket{s_0}{s_\delta}=0$, then $\braket{s_{\xi_1}}{s_{\xi_2}}=0$; conversely if $\braket{s_{\xi_1}}{s_{\xi_2}}=0$ then $\braket{s_\delta}{s_0}=0$ as well.
The other case, $\Delta \ni \delta = \xi_1 - \xi_2,$ is similar, except everything is complex conjugate.
\end{proof}

Next, we need to show that the matrix equation in Algorithm \ref{alg:construction} always has a solution. Fortunately we can always rely on the obvious classical solution:
\begin{lemma}
\label{lemma:feasible}
Equation \eqref{eq:abequation} always has at least one feasible solution: specifically, $\mathbf{b}_{equal}=(1/p,1/p,\dots,1/p).$
\end{lemma}
\begin{proof}
The first row of $A$, which checks normalization, is obviously satisfied.
As for the other rows, pick some $k' \in [1,|\Delta|].$ The corresponding rows $2k'-1, 2k'$ are all real; so let us consider adding row $2k'$ to $i$ times row $2k'-1$:
\begin{equation}
    \frac{1}{p}\sum_{j=0}^{p-1} \cos{(j\delta_{k'}\theta)} + i \sin {(j\delta_{k'}\theta)} = \frac{1}{p}\sum_{j=0}^{p-1} e^{ij\delta_{k'}\theta}.
\end{equation}
However, this right-hand side is just $1/p$ times the sum over the $\delta_{k'}$-th roots of unity of order $p$; thus, this is simply zero. Furthermore, since the rows $2k'-1, 2k'$ had all elements real, and they were the real and imaginary part of the above equation, it follows that those rows are also zero: thus, Eq.\eqref{eq:abequation} is satisfied.
\end{proof}

These lemmas are enough to prove that Algorithm \ref{alg:construction} always finds a valid solution to the classification problem:
\begin{theorem}
\label{thm:quantumexists}
The quantum agent constructed by Algorithm  \ref{alg:construction} solves $\taskMSC$.
\end{theorem}
\begin{proof}
First observe that if an agent can solve $\taskMCC$ then it can be used to solve $\taskMSC$,  and vice versa, by Lemmas \ref{lemma:tally1}-\ref{lemma:tally2}, so we prove this result in the context of solving $\taskMCC$.

Thus, we need only show that the resulting agent constructed by Algorithm \ref{alg:construction} solves $\taskMCC$. This requires us to show that for each $\Xi_j \in \Xi$, the resulting $S_j=\mathrm{Span}_{k \in \Xi_j}\{\ket{s_k}\}$ are orthogonal, and that the final measurement classifies them correctly.

Recall the unitary $U$ is constructed so that
\begin{equation}
    U\ket{\phi_j}=e^{\ii l_j \theta}\ket{\phi_j}.
\end{equation}
Computing $\braket{s_0}{s_{\delta_{k'}}}$ in this basis yields,
\begin{align} \label{eq:orthogonalityrequirements}
    \braket{s_0}{s_{\delta_{k'}}}&=\ematrix{s_0}{U^{\delta_{k'}}}{s_0}
    \nonumber\\
    &=\sum_j \sum_m \ematrix{\phi_j}{\sqrt{b_{l_j}}^* U^{\delta_{k'}}\sqrt{b_{l_m}}}{\phi_m}\nonumber\\
  &=\sum_{j=1}^{M}
e^{+i l_j\delta_{k'}\theta}b_{l_j}\nonumber\\
&=
\sum_{j'=0}^{p-1}
e^{+ij'\delta_{k'}\theta}b_{j'}.
\end{align}
where the last equation holds because $\exists j$ such that $l_j=j'$ precisely when $b_{j'}\neq 0.$

Meanwhile, Algorithm \ref{alg:construction} always returns a solution of $Ab= e_0$: by Lemma \ref{lemma:feasible}, a feasible solution always exists. Furthermore, there exists some invertible matrix $Q$ such that $QA=D$, and $Q\mathbf{e}_0=\mathbf{e}_0$. This implies  $A\mathbf{b}=\mathbf{e}_0$ if and only if $D\mathbf{b}=\mathbf{e}_0$. Consider  $2k'$-th row (where $k' \in \mathbb{Z}^+ $) of the matrix equation $D\mathbf{b}=\mathbf{e}_0$,
\begin{equation}
\label{eq:Drow}
    \sum_j  e^{+\ii j\delta_{k'}\theta}b_j = 0.
\end{equation}
This matches the left side of \eqref{eq:orthogonalityrequirements}. So if $A\mathbf{b}=\mathbf{e}_0$, then $\braket{s_0}{s_{\delta_{k'}}}=0$ for all $\delta_{k'}\in \Delta$.

Now by Lemma \ref{lemma:delta} if $\braket{s_0}{s_{\delta_{k'}}}=0$ for all $\delta_{k'}\in \Delta$ then the subspaces $S_j=\mathrm{Span}_{k \in \Xi_j}\{\ket{s_k}\}$ are mutually orthogonal to each other. Thus, we may find an orthonormal basis to each one, use it to construct a projection operator $P_j$ onto each one, and finally use them (and a rejection operator $P_0 = I-\sum_{j=1}^K P_j$) to build $\mathcal{M}$.

Finally, consider the result of the measurement $
\mathcal{M}$ on the state $U^L \ket{s_0}$ for $(L \mod p) = \xi \in \Xi_{j}$. As the spaces $S_j$ are orthogonal for distinct $j$, it follows that the result of measuring $\mathcal{M}$ on $U^L \ket {s_0}=\ket {s_\xi}$ is simply $j$. Thus, the agent $\mathcal{A}=(\mathcal{S},s_{init} = s_0,(U,\mathcal{M}))$ produced by Algorithm \ref{alg:construction} solves $\taskMCC$, completing the proof.
\end{proof}

In order to measure the memory of our construction,
we next introduce some definitions from linear programming \cite{forst2010}. Let $(FP)$ be the combination of the matrix equation $A\mathbf{x}=\mathbf{b}$ and a non-negativity constraint $\mathbf{x} \ge 0$; that is, $(FP)$ is a \textit{feasibility problem}, the problem of satisfying the constraints of some linear program. (We will not be concerned with the objective function.) We will require that $A$ has at most as many rows as columns, and that it is full rank.
Given $(FP)$, and a particular solution $\mathbf{x}$, we say that the \textit{support} $J_{\mathbf{x}}=\lbrace j | x_j\neq 0 \rbrace$ is the set of indices on which $x_j$ is nonzero. Let $J_{\mathbf{x}}(k)$ be the $k$-th element of $J_{\mathbf{x}}$. We then say that $\mathbf{x}$ is a \textit{basic feasible solution} to $(FP)$, a \textit{BFS}, if the submatrix of $A$ corresponding to $J_{\mathbf{x}}$ consists of linearly independent columns - that is, the submatrix defined by
\begin{equation}
    [A']_{ik}=[A]_{i J_{\mathbf{x}}(k)}
\end{equation}
consists of linearly independent columns. 

Next, if \textit{any} solution $\mathbf{x}$ to $(FP)$ exists, then a basic feasible solution $\mathbf{x'}$ also exists (proof taken from \cite{forst2010}).
\begin{lemma}
\label{lemma:bfs}
Suppose that there exists a solution $\mathbf{x}$ to $(FP)$  where $A$ is a $w \times p$ dimensional matrix of full rank with real entries, with $w \le p$. Then there exists a basic feasible solution $\mathbf{u}$ to $(FP);$ in particular, this basic feasible solution will have at most $w$ nonzero entries $u_j$.
\end{lemma}
\begin{proof}
If $\mathbf{x}$ is a basic feasible solution, we are done; so suppose otherwise. In particular, the submatrix $A'$ induced by the support $J_{\mathbf{x}}$ of $\mathbf{x}$ has linearly dependent columns; thus, the matrix equation $A'\mathbf{d'}=0$ has a nonzero solution $\mathbf{d'};$ extending back to the full space, we find that $A\mathbf{d}=0$ for some nonzero $\mathbf{d}$ such that $d_k=0$ if $x_k=0.$ We may assume without loss of generality that $d_k < 0$ for some $k$ (since otherwise, we can simply take the negative of $\mathbf{d},$ which also satisfies $A\mathbf{d}=0$).

Now, define the \textit{minimum ratio} as
\begin{equation}
    \alpha = \min_k \left\lbrace \left(\frac{x_k}{-d_k}\right): d_k<0 \right\rbrace.
\end{equation}
Let $\mathbf{x'}=\mathbf{x}+\alpha \mathbf{d}$. Then $\mathbf{x'} \ge 0$, $A\mathbf{x'}=\mathbf{b}$, and $J_{\mathbf{x'}} \subseteq J_{\mathbf{x}}$,
as $x_k=0 $ implies $ d_k=0$ and thus $x'_k=0$. In fact, $J_{\mathbf{x'}}$ is a \textit{strict} subset of $J_{\mathbf{x}}$,
since by construction $\exists k \in J_{\mathbf{x}}$ such that  $\mathbf{x'}_k=0$ and $ \mathbf{x}_k \neq 0$. Thus, this procedure has given us a new $\mathbf{x'}$ which has fewer nonzero components.

If $\mathbf{x'}$ is a basic feasible solution, we are done; else, we continue iterating this process, calculating a new $\mathbf{d},\alpha,$ and subtracting. Every iteration, $J_\mathbf{x'}$ becomes smaller, and thus the corresponding submatrix has fewer columns. Since any matrix consisting of a single non-zero column is linearly independent, this process must eventually terminate. Therefore, this process will always find a basic feasible solution.

Finally, if $A$ has $w$ rows, then a basic feasible solution $\mathbf{x}$ to $(FP)$ has $|J_{\mathbf{x}}|\le w$; otherwise, the submatrix $A'$ induced by $J_{\mathbf{x}}$ cannot have linearly independent columns. Thus, the basic feasible solution that our process finds has at most $w$ nonzero entries, completing the proof.
\end{proof}

Complementing this lemma, we have the following:
\begin{lemma}
\label{lemma:Dsquare}
Consider a BFS $\mathbf{b}$ to Eq.\eqref{eq:abequation} constructed by Algorithm \ref{alg:construction}. Then the number of nonzero entries $M = 2|\Delta|+1$.
\end{lemma}
\begin{proof}
By definition, the BFS must have $M \le 2|\Delta| + 1$; thus we must now prove $M \ge 2|\Delta| + 1$.
Let us proceed by contradiction.
Suppose that for $\mathbf{b}$  to  Eq. \eqref{eq:abequation}, $M < 2|\Delta| + 1$. Let $\mathbf{b}'$ equal $\mathbf{b}$ with all the zero rows deleted; thus, every remaining entry is nonzero. (Of course, because of row 0 of matrix $A$, $\mathbf{b}'$ must have at least one nonzero row.) Then delete every column $j$ of $A$ such that $b_j=0$, and row 0 of matrix $A$. The resulting matrix $A'$ still satisfies $A' \mathbf{b'} = 0$.

However, note that the remaining rows of $A'$ are in sine-cosine pairs; by applying elementary row operations to $A'$, one can obtain a (not necessarily square) sub-matrix $D'$ of  the DFT matrix; by Chebotarev's theorem, $D'$ is full rank (e.g. any $2|\Delta| \times 2|\Delta|$ submatrix is invertible) \cite{tao2003}; so $A'$ is full rank too.
(In particular, Chebotarev's theorem states: Let $p$ be a prime and $1 \le k \le p$, and $r_1, \ldots, r_k$ be distinct elements of $\mathbb{Z}_p$. Similarly let $\mu_1,\ldots \mu_k$ also be distinct elements of $\mathbb{Z}_p$. Consider a $k\times k$ sub-matrix of the DFT matrix  of order $p$, with entries
$[e^{2 \pi i r_l\mu_j /p}]_{jl}$ where $1\le j,l \le k$. This matrix has non-zero determinant  (see \cite{tao2003} for full statement and proof).)

Because $A'$ is full rank (i.e. rank $M$), the linear system $A' \mathbf{b'} = 0$ contains $M$ linearly independent constraints, and $M \le 2|\Delta|$ variables. The trivial solution obviously solves this equation, but the system must therefore not have any other solutions.
Therefore, there is only one possible solution to this new matrix equation: $\mathbf{b'}=0$, a contradiction.
Thus, $M = 2|\Delta| + 1$.
\end{proof}

Now, consider a solution to equation \eqref{eq:abequation} with $M=2|\Delta|+1$ nonzero entries relative to the eigenbasis. By construction, the set of states $\setpStates=\{\ket{s_j}=U^j \ket{s_0}\}$ is contained in a Hilbert space of dimension $M$, so the memory cost $\dimHacc\le M$. We now prove that this upper bound is also tight.

\begin{lemma}
\label{lemma:rankisdim}
Consider the set of states $S=\{\ket{s_j}=U^j \ket{s_0}\}$ constructed by Algorithm \ref{alg:construction}. Then the memory cost $\dimHacc = M.$
\end{lemma}
\begin{proof}
First, by construction, the span of the set of states $S=\{\ket{s_j}=U^j \ket{s_0}\} \subseteq \mathcal{H}$ constructed by Algorithm \ref{alg:construction} clearly has dimension $\dimHacc\le M = 2|\Delta|+1$, where the latter equality follows from Lemma \ref{lemma:Dsquare}. We now prove $\dimHacc = M$.

Now, consider the matrix $A$ in equation \eqref{eq:Amatrix}. By Lemma \ref{lemma:Dsquare}, $M = 2|\Delta|+1.$ It is therefore easy to see that $A$ can be transformed by elementary row operations into an $M \times p $ sub-matrix of the Discrete Fourier Transform matrix of order $p$; call this matrix $D$. Roughly, $D$ is obtained by `pairing up' rows of $A$ with matching $k'$. More precisely,
\begin{equation}
\label{eq:DFT}
[D]_{kj}=\begin{cases}
        e^{+\ii j\delta_{k'}\theta} & k=2k' \textrm{ for } k'\in\mathbb{Z},\\
        e^{-\ii j\delta_{k'}\theta} & k=2k'-1 \textrm{ for } k'\in\mathbb{Z}
        \end{cases}
\end{equation}
and $\delta_{k'}$ is the $k'$-th element of $\Delta$, where by  convention $\delta_0 = 0.$ From each pair of rows of $D$  (indexed by $k =2k'$ and $k = 2k' -1$ for the same $k' \in \mathbb{Z}$), we get two equations: one for the real part of the constraint, and one for the imaginary.  These real and imaginary `part' equations are exactly the equations expressed by the corresponding rows of $A$.  As we require the solution to $A\mathbf{b} = \mathbf{e}_0$ to be elementwise real (and positive), writing the equations in terms of $A$ vs $D$ is implicitly like separating the real and imaginary parts of the equation vs grouping them together.

Now, let us use the $b_j$ solving Eq. \eqref{eq:abequation} to define the following $(p\times M)$ matrix:
\begin{equation}
    [B]_{jj'}=\begin{cases}
        \sqrt{b_{l_{j'}}} & \textrm{if } j=l_{j'}\\
        0 & \textrm{otherwise}
    \end{cases}
\end{equation}
where $l_j$ is the $j$-th nonzero element of the vector $\mathbf{b}$.

After multiplying $D$ by $B$, we get a new $M\times M$ matrix $DB$ whose rows are $\ket{s_{n}}$\ in the eigenbasis:
\begin{equation}
    [D B]_{kj'}=\begin{cases}
         \sqrt{b_{l_{j'}}} e^{+\ii l_{j'}\delta_{k'}\theta} & k=2k' \textrm{ for } k'\in\mathbb{Z},\\
         \sqrt{b_{l_{j'}}} e^{-\ii l_{j'}\delta_{k'}\theta} & k=2k'-1 \textrm{ for } k'\in\mathbb{Z};
        \end{cases}
\end{equation}
that is, the $k$-th row of this matrix is precisely $\braket{\phi_{j'}}{s_{\delta_{k'}}}$ if $k=2k'$ and $\braket{\phi_{j'}}{s_{p-\delta_{k'}}}$ if $k=2k'-1$, where $\{\ket{\phi_j}\}$ is the orthonormal basis of $\mathcal{H}$ found in Algorithm \ref{alg:construction}.

Now, suppose we rescaled each column of $DB$, dividing by $\sqrt{b_{l_{j'}}}$. We get an $M\times M$-dimensional matrix $D'$, which is a square sub-matrix of the DFT matrix of order $p$.
Chebotarev's theorem on roots of unity states that all square sub-matrices of the DFT matrix are invertible \cite{tao2003}; thus, so is $D'$.

Thus, we have shown $D'$ is invertible, and in particular, its rows are linearly independent.
Rescaling the columns of a matrix does not affect its rank. However, by applying this rescaling, we can transform $D'$ into $DB$. Thus, the vectors $\{U^n \ket{s_0} |\,n = 0 \textrm{ or } n \in \Delta \textrm{ or } p-n \in \Delta\}$, the rows of $DB$, are linearly independent. Therefore, our agent has memory cost $\dimHacc \ge M.$

Since we have previously shown $\dimHacc \le M$, it follows that $\dimHacc = M$, completing the proof.
\end{proof}

We are now ready to prove the main theorems of this section:

\begin{theorem}
\label{thm:quantumdim}
The quantum agent constructed by Algorithm \ref{alg:construction} has dimension $M =2|\Delta| +1$.
\end{theorem}

\begin{proof}
We observe that Eq. \eqref{eq:abequation}  represents a set of linear equality constraints, and the requirement $\mathbf{b}\ge 0$ is an inequality constraint. Together these define a linear program $P$. Furthermore, all elements of $A$ are real.

First, by Lemma \ref{lemma:feasible}, a feasible solution always exists. We must then show that  $A$ in Eq. \eqref{eq:Amatrix} is full rank, and thus Lemma \ref{lemma:bfs} applies: that is, there is always a solution to Eq. \eqref{eq:abequation} with at most $M = 2|\Delta| + 1$ non-zero entries.

By construction, $A$ has at least as many columns as rows. It is also easy to show that it is related by row operations to a submatrix of the Discrete Fourier Transform matrix of order $p$.  Specifically, there exists some invertible row-operation matrix $Q$ such that $QA=D$, where $D$ is as defined in Equation \eqref{eq:DFT}. Since $D$ is a submatrix of the DFT matrix of order $p$ prime, Chebotarev's theorem on roots of unity implies $D$ is full rank \cite{tao2003}. Therefore, $A$ is also full rank. Thus, by Lemma \ref{lemma:bfs}, a \textit{basic feasible} solution to $P$ exists, with $M \le 2|\Delta| +1$ nonzero entries; by Lemma \ref{lemma:Dsquare}, this bound is tight, so $M = 2|\Delta| + 1.$

By Lemmas \ref{lemma:rankisdim}, the state space generated has dimension $\dimHacc = M;$ this implies the memory cost of the construction recovered by Algorithm \ref{alg:construction} is $M = 2|\Delta| +1$.
\end{proof}

We conclude by proving the simple upper bound mentioned at the beginning of this section.
\begin{lemma}
\label{lemma:quantumsimple}
    There exists a quantum agent $\agent$ solving $\taskMSC$ with memory cost at most $(\sum_j|\setMods_j|)^2$.
\end{lemma}
\begin{proof}
    The size of the difference set $\Delta$ can be upper bounded by the number of pairs of distinct elements in $\bigcup_j\setMods_j;$ that is, $|\Delta| \le \left(\sum_j|\setMods_j|\right)\left(\sum_j|\setMods_j|-1\right)/2$. By Thm.\ref{thm:quantumdim}, there exists a quantum agent with memory cost $M \le 2|\Delta|+1;$ so,
    \begin{align}
        M &\le 2|\Delta| + 1\\
        &\le \left(\sum_j|\setMods_j|\right)\left(\sum_j|\setMods_j|-1\right)+1\\
        &\le \left(\sum_j|\setMods_j|\right)\left(\sum_j|\setMods_j|-1\right)+\left(\sum_j|\setMods_j|\right)\\
        &= \left(\sum_j|\setMods_j|\right)^2,
    \end{align}
    proving our simple upper bound.
\end{proof}

\section{Optimality of Unitary Agents}
\label{supp:sec:unitary}

We prove the optimality of the quantum agent constructed in the previous section by proving any quantum construction for solving \taskMSC \, needs memory dimension of at least $M_Q = 2 |\Delta| +1$.

\begin{lemma}
\label{lemma:pplusone}
Suppose an agent $\mathcal{A}=(\mathcal{S},s_{init} =s_0,\mathcal{T})$ solves \taskMCC, where $\mathcal{T}=(\fnTransOne,\mathcal{M})$. Then for any positive integer $n$, $\mathcal{A}'=(\mathcal{S}',s'_{init} = s_0,\mathcal{T'})$, with $\mathcal{T'}=(\fnTransOne^{np+1},\mathcal{M})$, also solves $\taskMCC$, with memory cost no worse than $\mathcal{A}$; in particular, $\mathcal{S}' \subseteq \mathcal{S}.$
\end{lemma}
\begin{proof}
Consider an arbitrary input string $\strInput^{(1)}$ in $\setData$ that is equivalent to a sequence of $\kappa$ ones, followed by $\Omega$. After processing this word, the agent $\mathcal{A}$ is in the state $s^{(1)}=\fnTransOne^\kappa (s_0),$ after which we apply the measurement $\mathcal{M}$, which results in label $j=\fnClass{\strInput^{(1)}}$. Now, suppose we instead input the word $\strInput^{(2)}$ consisting of $(np+1)\kappa$ 1's: the resulting state would then be $s^{(2)}=\fnTransOne^{(np+1)\kappa} (s_0);$ and since $(np+1)\kappa \equiv \kappa \mod p$, and $\mathcal{A}$ solves $\taskMCC$, the measurement $\mathcal{M}$ on $s^{(2)}$ would also result in label $j.$

Now suppose we were to instead use $\mathcal{A}'=(\mathcal{S}',s_0,\mathcal{T}')$, where $\mathcal{T}'=(\fnTransOne^{np+1},\mathcal{M}).$ On input $\strInput$, the agent's pre-measurement state would again be $s^{(2)}=G^{(np+1)\kappa}(s_0).$ Since $\mathcal{A}'$ still uses the same measurement $\mathcal{M},$ by the argument in the previous paragraph, this measurement $\mathcal{M}$ still yields label $j$. Now because $\strInput^{(1)}\in \setData$ was arbitrary, we have demonstrated $\mathcal{A}'$ also solves $\taskMCC$.

As for the memory cost, the accessible space $\mathcal{H}(\agent')$
of $\agent'$ is
\begin{equation}
    \mathcal{H}(\agent') = \mathrm{span} \bigcup_\kappa \mathrm{supp}(\fnTransOne^{\kappa(np+1)}(s_0)),
\end{equation}
which is clearly a subset of the accessible space of $\agent$,
\begin{equation}
    \mathcal{H}(\agent) = \mathrm{span} \bigcup_\kappa \mathrm{supp}(\fnTransOne^{\kappa}(s_0)).
\end{equation}
Therefore, the memory cost of $\mathcal{A}'$ is no worse than that of $\mathcal{A}$, completing the proof.
\end{proof}


Now, as $\fnTransOne$ is a map between density matrices, we may also consider its properties as a \textit{linear} map. Let us use the usual definition of eigenvectors:

\begin{definition}
A matrix $\rho$ is an \textit{eigenvector} of the map $\fnTransOne$ with eigenvalue $\lambda$ if and only if
\begin{equation}
    \fnTransOne(\rho) = \lambda \rho.
\end{equation}
\end{definition}
Note that we do \textit{not} demand that the eigenvector is a \textit{density} matrix: in particular, it need not be positive semidefinite nor unit trace.

Now, consider the matrix $\hat{\fnTransOne}: \mathbb{C}^d \rightarrow \mathbb{C}^d$. While it is not necessarily Hermitian (i.e. the spectral theorem does not apply), we can use its eigenvectors to find another canonical form. There exists a matrix $Q$ such that $\hat{\fnTransOne}=  Q J Q^{-1}$, where the \textit{Jordan canonical form} $J = J_{\lambda_1}\oplus J_{\lambda_2}\oplus \dots$ is block diagonal, such that each \textit{Jordan block}
\begin{equation}
    J_{\lambda_j}=\begin{bmatrix}
    \lambda_j & 1 & 0 & \dots & 0\\
    0 & \lambda_j & 1 & \dots & 0\\
    0 & 0 & \lambda_j & \dots & 0\\
    \vdots \\
    0 & 0 & 0 & \dots & \lambda_j
    \end{bmatrix}
\end{equation}
is an off-diagonal $d_{\lambda_j} \times d_{\lambda_j}$ square matrix, with the associated eigenvalue $\lambda_j$ on the diagonal, 1's directly above the diagonal, and zeroes everywhere else. (We will not assume the $\lambda_j$ are distinct.) In particular, if $\hat{\fnTransOne}$ is \textit{diagonalizable}, all Jordan blocks will be $1 \times 1$, although we will not assume this is the case.

Additionally, following \cite{wolf2012}, we will split each block into a projection $\Pi_j$ and a nilpotent part $N_j$. Therefore,
\begin{align}
\label{eq:projnil}
    \hat{\fnTransOne}&=\sum_j \left(\lambda_j \Pi_j + N_j\right), \\
    N_j^{d_j}&=0, \,\, \, N_j \Pi_j = \Pi_j N_j = N_j, \\
    \Pi_j \Pi_k &= \delta_{jk} \Pi_j,\,\,\, \tr \Pi_j = d_j, \\
    \sum_j \Pi_j &= \mathbbm{1}.
\end{align}
In effect, $\lambda_j\Pi_j$ is the diagonal part of the $j$-th Jordan block, while $N_j$ is the off-diagonal part.
Thus, $\Pi_j$ represents the \textit{spectral projection} into the generalized eigenspace of $\lambda_j.$

Now we have defined $\fnTransOne$ to be a CPTP map. When $\fnTransOne$ is CPTP, it can be shown that  all the eigenvalues of $\fnTransOne$ have magnitude at most 1 (for proof see \cite{watrous2018} Prop 4.26.).

We next prove constraints on the eigenspace of $\fnTransOne$, Let us define the following:

\begin{definition}
The \textit{peripheral spectrum} of $\fnTransOne: \mathcal{M}_M(\mathbb{C}) \rightarrow \mathcal{M}_M(\mathbb{C})$ is the set of eigenvalues of unit magnitude, i.e. $|\lambda|=1$.
\end{definition}

 For a CPTP map $\fnTransOne$ we will show the peripheral spectrum consists of only $1\times 1 $ Jordan blocks. To establish this we directly adapt Proposition 6.2 from \cite{wolf2012}, and include the proof from the same source below:
\begin{lemma}
\label{lemma:peripheraldiag}
Let $\lambda$ be an eigenvalue of a positive trace-preserving map $\fnTransOne :\mathcal{M}_M(\mathbb{C}) \rightarrow \mathcal{M}_M(\mathbb{C})$, such that $|\lambda|=1$. Then all Jordan blocks corresponding to $\lambda$ are 1 $\times$ 1.
\end{lemma}
\begin{proof}
Suppose $A, B \in \mathcal{M}_M (\mathbb{C})$, and let $\fnTransOne: \mathcal{M}_M (\mathbb{C}) \rightarrow \mathcal{M}_M (\mathbb{C})$ be a linear map.  Then the expression $\tr [A\fnTransOne(B)]$ can be upper bounded with respect to the norms of $A$ and $B$ and the linear dimension $d$ of these matrices. In particular, every matrix can be expressed as a linear combination of positive matrices, and for \textit{positive} $A, B$ we have $\tr [A\fnTransOne(B)] \le ||A||_\infty ||B||_\infty \tr [\mathbbm{1} \fnTransOne (\mathbbm{1})]=M ||A||_\infty ||B||_\infty$. A  similar result holds for $\fnTransOne^n$ for all $n \in \mathbb{N}$ (we say $\fnTransOne^n$ is uniformly bounded).

However, suppose $\hat{\fnTransOne}$ has a nontrivial (i.e. larger than $1 \times 1$) Jordan block $J_{\lambda}$ corresponding to $\lambda$.
A brief calculation shows that the top row of the $n$-th power of that block is
\begin{equation}\label{eq:nthpowerjordanblock}
    [J_{\lambda}^n]_{1k}=\lambda^{n-k+1} \binom{n}{k-1}, \quad k=1,..., d_{\lambda}
\end{equation}
which grows without bound due to the expression for  the binomial coefficients. However $\fnTransOne^n$ is uniformly bounded; thus, all Jordan blocks $J_{\lambda}$ must be $1 \times 1$.
\end{proof}
As a result of this lemma, we find that the Jordan blocks of the peripheral spectrum are all trivial; thus, the `peripheral part' of $\fnTransOne$ is diagonalizable.

We are now ready to prove the first major result of this section, by modifying a construction from \citep{wolf2012}.
\begin{lemma}
\label{lemma:quantum01b}
Suppose an agent $\mathcal{A}=(\mathcal{S},s_{init} =s_0,\mathcal{T})$ solves \taskMCC, where $\mathcal{T}=(\fnTransOne,\mathcal{M})$. Let
\begin{equation}
\fnTransOne_{\varphi}=\sum_{j:|\lambda_j|=1}\lambda_j \Pi_j.
\end{equation}
Let $\mathcal{T}'=(\fnTransOne_{\varphi}, \mathcal{M}).$
Then $\mathcal{A}'=(\mathcal{S},s'_{init} = s_0,\mathcal{T}')$ also solves \taskMCC, with memory cost no worse than $\mathcal{A}$. Furthermore, $\fnTransOne_{\varphi}$ is diagonalizable, and all its nonzero eigenvalues have unit magnitude.
\end{lemma}
\begin{proof}
    Consider the peripheral spectrum of $\fnTransOne$, i.e. the eigenvalues of  $\fnTransOne$, satisfying $ |\lambda_j|=1.$ These peripheral spectrum values also trivially satisfy $|\lambda_j^p| = 1$. Construct  the set of values $\{\lambda_j^p\}$, for $\lambda_j$ belonging to the peripheral eigenspectrum of $\fnTransOne$.
    By Dirichlet's approximation theorem applied to this set $\{\lambda_j^p\}$, for all $\epsilon > 0$ there exists $n \in \mathbb{N}$ such that $|\lambda_j^{np}-1|\le \epsilon$ for all $\lambda_j$. Thus, we may construct an ascending sequence $n_i \rightarrow \infty$ such that $\lambda_j^{n_i p} \rightarrow 1$, and this limit is approached uniformly in $j$.

    Let us now consider the block decomposition of $\hat{\fnTransOne}^{n_ip +1}.$ Using the Jordan normal form, and the decomposition in Eq.\eqref{eq:projnil}, we must have
    \begin{equation}
        \hat{\fnTransOne}^{n_i p+1} = \sum_j (\lambda_j \Pi_j + N_j)^{n_i p+1}.
    \end{equation}

    However, since $\fnTransOne$ is a trace-preserving CPTP map, we have  $|\lambda_j| \le 1$ for all $j$ (for a proof see \cite{watrous2018} Prop 4.26.). Clearly, if this inequality is strict such that $|\lambda_j| < 1$, then in the limit  $n \rightarrow \infty$, the blocks $(\lambda_j \Pi_j + N_j)^{np+1} \rightarrow 0$. Meanwhile for the peripheral spectrum eigenspaces, by Lemma \ref{lemma:peripheraldiag}, $|\lambda_j|=1$ implies $N_j=0.$ In conclusion,
    \begin{align}
    \label{limit_G}
          \lim_{i \rightarrow \infty} \fnTransOne^{n_i p + 1} &= \lim_{i \rightarrow \infty} \sum_{j:|\lambda_j|=1}\lambda_j^{n_i p + 1} \Pi_j  \notag \\& \rightarrow  \sum_{j:|\lambda_j|=1} \lambda_j \Pi_j = \fnTransOne_\varphi.
    \end{align}

    Meanwhile for each $n_i \in \mathbb{N}$, we can use Lemma \ref{lemma:pplusone} to argue that $\fnTransOne^{n_ip+1}$ will solve \taskMCC.

    Thus there exists a sequence of operators $\fnTransOne^{n_i p + 1}$ that solve \taskMCC, and that also  converge to $\fnTransOne_\varphi$. Since the space of CPTP maps is closed, $\fnTransOne_\varphi$ is also CPTP. In the next paragraph, we use this fact, in conjunction with continuity of  measurement outcome probabilities in terms of the operators $\fnTransOne, \mathcal{M}$, to show that $\mathcal{T}'$ also solves \taskMCC.

    Consider an arbitrary classification $j$, and an arbitrary word $\strInput$ such that $\fnClass{\strInput}=j$. Let $\mathrm{Pr}_{\tilde{\fnTransOne}, M}(j,\strInput)$ be the probability of measuring $j$ after input $\strInput$ using transition map $\tilde{T} =(\tilde{\fnTransOne}, M).$ This probability is clearly continuous in $\tilde{\fnTransOne}$. In particular, since
    \begin{equation}
        \fnTransOne_\varphi=\lim_{i \rightarrow \infty} \fnTransOne^{n_i p + 1},
    \end{equation}
    then
    \begin{equation}
        \mathrm{Pr}_{\fnTransOne_\varphi, M}(j,\strInput)=\lim_{i \rightarrow \infty} \mathrm{Pr}_{\fnTransOne^{n_i p + 1},M}(j|\strInput).
    \end{equation}

    However, recall that $\fnTransOne^{n_i p + 1}$ solves $\taskMCC$. Since $\strInput \in \Gamma_j$, it immediately follows that $\mathrm{Pr}_{\fnTransOne^{n_i p + 1}}(j,\strInput)=1$ for \textit{all} $i$. Thus, $\mathrm{Pr}_{\fnTransOne'}(j,\strInput)=1$, and so $ \fnTransOne_\varphi$ classifies $\strInput\in\Gamma_j$ correctly. Finally, since $j$ and $\strInput$ were arbitrary, $ \fnTransOne_\varphi$ therefore solves the problem $\taskMCC$.

For $\kappa=0$, the state $s_0$ is already reachable under $\agent$.
For every integer $\kappa \geq 1$,
\begin{equation}
\fnTransOne_\varphi^\kappa(s_0)
=
\lim_{i\to\infty}G^{n_i p+\kappa}(s_0),
\end{equation}
where each term on the right-hand side is reachable under $\agent$, i.e. in $\mathcal {R}(\mathcal{A})$. For \textbf{any} converging sequence of PSD states $\rho_i \rightarrow \rho$, $\operatorname{supp}(\rho) \subseteq \operatorname{span}\bigcup_i \operatorname{supp}(\rho_i);$ and so, $\operatorname{supp}\fnTransOne_\varphi^\kappa(s_0) \subseteq \mathcal{H}(\mathcal{A})$. However, $\kappa$ was arbitrary; and thus $\mathcal{H}(\mathcal{A'}) \subseteq \mathcal{H}(\mathcal{A})$ Thus the memory cost of $\agent'$ is no worse than $\agent$.

\end{proof}

On its own, Lemma \ref{lemma:quantum01b} means that we can assume $\hat{\fnTransOne}$ is diagonalizable. However we have not proven that the eigenstates of the operator in Eq. \eqref{limit_G} are pure.

However, we can apply the results of \cite{wolf2010} to solve this problem. Define
\begin{equation}
    \mathcal{X}_\fnTransOne = \textrm{span}\lbrace X\in \mathcal{M}_M (\mathbb{C}) | \exists \phi \in \mathbb{R} : \fnTransOne(X)=e^{\ii \phi}X \rbrace.
\end{equation}
That is, $\mathcal{X}_\fnTransOne$ is the span of all eigenvectors with eigenvalues of magnitude one. By Lemma \ref{lemma:quantum01b}, we can assume \textit{all} nonzero eigenvalues of $\fnTransOne$ have magnitude one, so $\mathcal{X}_\fnTransOne$ is actually the row space or coimage of $\fnTransOne$; in particular, $\fnTransOne$ is specified by its behaviour on $\mathcal{X}_\fnTransOne$.

We can also identify extremal states \footnote{These are called ``pure states in $\mathcal{X}_\fnTransOne$" in \cite{wolf2010}, but for clarity we will not use that term here because they need not correspond to pure states in the density matrix picture.} inside this subspace, following \cite{wolf2010}:
\begin{definition}\label{def:peripheral_pure_state}
An \textit{extremal state in $\mathcal{X}_\fnTransOne$} is a state which cannot be expressed as a nontrivial convex combination of states in $\mathcal{X}_\fnTransOne$.
\end{definition}

With these definitions, we may proceed with the main result of this section:
\begin{theorem}
\label{thm:unitary}
Consider an agent $\mathcal{A}=(\mathcal{S},s_{init} = s_0,\mathcal{T})$ solving $\taskMCC$.
Then there exists a new agent $\mathcal{A}'=(\mathcal{S}',s'_{init} = s'_0,\mathcal{T}')$ also solving \taskMCC, where $\mathcal{T}'=(\fnTransOne',\mathcal{M}')$, and $\fnTransOne'$ is a unitary (i.e. Kraus rank 1) quantum process. Furthermore $\mathcal{A}'$ has memory cost no larger than that of $\mathcal{A}$.
\end{theorem}
\begin{proof}
Let $\mathcal{T}=(\fnTransOne,\mathcal{M}).$
By Lemma \ref{lemma:quantum01b}, we may replace $\fnTransOne$ by a new function $\fnTransOne_\varphi$ which is diagonalizable, and such that \textit{all} the nonzero eigenvalues of $\fnTransOne_\varphi$ are `phases', i.e. of magnitude one.
Therefore, $\fnTransOne_\varphi$ is specified by its behaviour on $\mathcal{X}_{\fnTransOne_\varphi}$.
Then by theorem 8 of \cite{wolf2010} we can decompose the Hilbert space
\begin{equation}
    \mathcal{H}=\mathbb{C}^M=\mathcal{H}_0 \oplus \bigoplus_{j=1}^{J}\mathcal{H}_{j},
\end{equation}
with $\mathcal{H}_j=\mathcal{H}_{j,1}\otimes \mathcal{H}_{j,2}$ and $M = \sum_{j=0}^J {\rm dim}(\mathcal{H}_j)$. Furthermore we can find density matrices $\rho_j$ acting on $\mathcal{H}_{j,2}$, such that
\begin{equation}
    \mathcal{X}_{\fnTransOne_\varphi}=0 \oplus \bigoplus_{j=1}^J \mathcal{M}_{d_j}\otimes \rho_j
\end{equation}
where each $\mathcal{M}_{d_j}$ is a full matrix algebra (i.e. $\mathcal{M}_{d_j}=B(\mathcal{H}_{j,1})$ \cite{lindblad1999}) over the Hilbert space $\mathcal{H}_{j,1}$ of dimension $d_j=\textrm{dim}(\mathcal{H}_{j,1})$. This implies that for all $X\in\mathcal{X}_{\fnTransOne_\varphi},$ there exist matrices $\chi_j \in B(\mathcal{H}_{j,1})$ such that
\begin{equation}
\label{eq:xrhodecomp}
    X=0 \oplus \bigoplus_{j=1}^J \chi_j\otimes \rho_j.
\end{equation}

Next, we note that $\fnTransOne_\varphi$ is invertible in $\mathcal{X}_{\fnTransOne_\varphi}$. That is, there exists a map $\fnTransOne_\varphi^{-1}$ such that for all $X \in \mathcal{X}_{\fnTransOne_\varphi},$ $\fnTransOne_\varphi^{-1} \fnTransOne_\varphi (X) = X.$
We may construct $\fnTransOne_\varphi^{-1}$ using the methods found in Lemma \ref{lemma:quantum01b}, except this time, instead of constructing the limit $\fnTransOne_{\varphi}=\lim_{i\rightarrow \infty}\fnTransOne_\varphi^{n_i p + 1}$, we construct $\fnTransOne_\varphi^{-1}=\lim_{i\rightarrow \infty}\fnTransOne_\varphi^{n_i p - 1}$.

We now lift a result from the proof of Theorem 8 in \cite{wolf2010}: the existence of this inverse $\fnTransOne_\varphi^{-1}$ implies that $\fnTransOne_\varphi$ (and thus $\fnTransOne_\varphi^{-1}$) map extremal states in $\mathcal{X}_{\fnTransOne_\varphi}$ to other extremal states. In particular we give a brief sketch of the proof here. Suppose that $\fnTransOne_\varphi$ does not map extremal states to extremal states: that is, suppose there exists some pure state $s = \chi \otimes \rho_j \in \mathcal{X}_{\fnTransOne_\varphi}$ which gets mapped to $\fnTransOne_\varphi(s) = \alpha_1 x_1 \otimes \rho_1 + \alpha_2 x_2 \otimes \rho_2$, $\alpha_1 + \alpha_2 = 1.$ Then applying $\fnTransOne_\varphi^{-1}$, we find that $s$ may be expressed as the nontrivial convex combination $s = \alpha_1 \fnTransOne_\varphi^{-1}(\chi_1 \otimes \rho_1) + \alpha_2 \fnTransOne_\varphi^{-1}(\chi_2 \otimes \rho_2)$, which is a contradiction. (Note that as $\fnTransOne_\varphi^{-1}$ is itself invertible, $\fnTransOne_\varphi^{-1}(\chi_1 \otimes \rho_1) \neq \fnTransOne_\varphi^{-1}(\chi_2 \otimes \rho_2)$.) Thus $\fnTransOne_\varphi$ and $\fnTransOne_\varphi^{-1}$ map extremal states to extremal states.

Now, any state in $\mathcal{X}_{\fnTransOne_\varphi}$ may be written as in Eq. \eqref{eq:xrhodecomp}. Clearly, the projection of $X$ onto each $j$-space is in $\mathcal{X}_{\fnTransOne_\varphi}$. It immediately follows that for any \textit{extremal} state in $\mathcal{X}_{\fnTransOne_\varphi}$, $\chi_j$ must vanish for all but one $j$. In fact, we can go further: for every extremal state in $\mathcal{X}_{\fnTransOne_\varphi}$, the nonzero $\chi_j$ is extremal in $B(\mathcal{H}_{j,1})$; that is, $\chi_j$ is a rank-one projection, and is a pure state in the usual sense \cite{wolf2010}.

First, we may assume $s_0$ is contained within $\mathcal{X}_{\fnTransOne_\varphi}$. Suppose otherwise: since prepending $p$ ones to a string should not affect its classification, we can construct a new agent $\agent'' = (\mathcal{S}', \fnTransOne_\varphi^{p}(s_0), \mathcal{T}'),$ and the resulting agent will still solve $\taskMCC$. Since $\fnTransOne_\varphi$ has all eigenvalues magnitude either 0 or 1, the new initial state $\fnTransOne_\varphi^p(s_0)$ must be entirely contained within $\mathcal{X}_{\fnTransOne_\varphi}$. Thus assuming $s_0 \in \mathcal{X}_{\fnTransOne_\varphi}$ is justified.

Now, suppose $s_0$ is not extremal in $\mathcal{X}_{\fnTransOne_\varphi}$. In particular, it has a nontrivial expression as a convex combination of extremal states, a ``decomposition", and each extremal state in the decomposition can be written as $\chi \otimes \rho_j$ for some $j$ and some pure state $\chi \in B(\mathcal{H}_{j,1})$. Let us argue similarly to Lemma \ref{lemma:purity}. We first assume, using Lemma \ref{lemma:zero}, that our original agent $\agent$ emits 0 every timestep except when $x_T=\Omega$. In that case, for any input word in the classes we must classify, the pointwise loss of $\agent$ (i.e. the probability of the measurement on a possibly \textit{mixed} memory state getting the wrong classification), expressed as a function of the initial state as in Eq.\eqref{eq:affineloss}, is affine in the initial state $s_0$; in particular, if the loss is zero for $s_0$, it must also be zero for every extremal state $\chi \otimes \rho_j$ in its decomposition, and so an agent initialized in state $\chi \otimes \rho_j$ \textit{also} solves $\taskMCC$. Thus, we may assume that $s_0$ is extremal: that is, as a matrix,
\begin{equation}
    s_0=\chi \otimes \rho_{\hat{j}}
\end{equation}
for some rank-one projection $\chi \in B(\mathcal{H}_{j,1})$ and some particular $\hat{j}$.

Again by Theorem 8 in \cite{wolf2010}, there exists unitaries $U_j$, and a permutation $\pi$ permuting between indices $1,\dots,J$, such that for all $\psi\in \mathcal{X}_{\fnTransOne_\varphi}$,
\begin{equation}
    \fnTransOne_\varphi(\psi)=0 \oplus \bigoplus_{j=1}^J U_j \chi_{\pi(j)} U_j^\dagger \otimes \rho_j.
\end{equation}
That is, the permutation $\pi$ and the unitaries $U_j$ do not depend on $\psi$. Note that $\pi$ only permutes between Hilbert spaces of matching dimension.

Now, let $\nu$ be the length of the cycle containing $\hat{j}$ in the cyclic decomposition of $\pi$; that is, the smallest $n$ such that $\pi^n(\hat{j})=\hat{j}.$ Since this is the smallest such $n$, $B(\mathcal{H}_{\pi^n(\hat{j}),1}) \otimes \rho_{\pi^n (\hat{j})} \subseteq \mathcal{X}_{\fnTransOne_\varphi}$ must be distinct for $n=0,1,2,\ldots,\nu-1$. Counting dimensions, we must have $d_{\hat{j}} \nu \le M,$ (as $\pi$ can only permute between Hilbert spaces of matching dimension).  In particular this implies $\nu \le M$.

Suppose our agent $\mathcal{A}$ has memory cost $M\ge p$: we can then always build a \textit{classical} agent $\mathcal{A}'$ with memory cost $p$. Now, the classical agent's update map $\fnTransOne_\varphi'$ is unitary, and the memory cost of $\mathcal{A}'$ is no larger than the memory cost of $\mathcal{A}$. This implies our theorem is automatically true in the case $M\ge p$. So, in what follows, we assume $M < p$, and in particular $\nu < p$.

Since $\nu < p$, it is coprime with $p$. In particular, by the Chinese remainder theorem, there exists some $N$ such that
\begin{align}
    &N \,\,{\rm mod }\,\, \nu \equiv 0 \\
    &N \,\,{\rm mod }\,\, p \equiv 1.
\end{align}
Suppose we apply $\fnTransOne_\varphi$ to the extremal state $s_0 = \chi \otimes \rho_{\hat{j}}$ \textit{once}. Since the only nonzero component of $s_0$ is inside $B(\mathcal{H}_{\hat{j},1}) \otimes \rho_{\hat{j}}$, we find that
\begin{equation}
    \fnTransOne_\varphi(s_0) = U_{\pi^{-1}(\hat{j})} \chi U^\dagger_{\pi^{-1}(\hat{j})} \otimes \rho_{\pi^{-1}(\hat{j})}.
\end{equation}

Now, consider $\fnTransOne_\varphi^N (s_0)=\fnTransOne_\varphi^N(\chi\otimes \rho_{\hat{j}})$.  Since $\Nmodx{\nu}\equiv 0$, $\pi^N (\hat{j})=\hat{j},$ and thus $\fnTransOne_\varphi^N$ maps $B(\mathcal{H}_{\hat{j},1})\otimes \rho_{\hat{j}}$ to \textit{itself.} Therefore,
\begin{align}
    \fnTransOne_\varphi^N &(s_0)\nonumber\\
    &=(U_{\hat{j}} U_{\pi(\hat{j})}\dots U_{\pi^{N-1}(\hat{j})} \chi U_{\pi^{N-1}(\hat{j})}^\dagger \dots U_{\pi(\hat{j})}^\dagger U_{\hat{j}}^\dagger)\otimes \rho_{\hat{j}}.
\end{align}
Note that here we have used the fact that $\pi^N (\hat{j})=\hat{j},$ and thus $\pi^{-1}(\hat{j}) = \pi^{N-1} (\hat{j})$.

However, by Lemma \ref{lemma:pplusone}, since $N\,\, {\rm mod}\,\, p \equiv 1$, $\mathcal{A}_N=(\mathcal{S},s_{init} = s_0,(\fnTransOne_\varphi^N,\mathcal{M}))$  should lead to a new agent that solves \taskMCC. In particular, for \textit{any} $k$, $\fnTransOne_\varphi^{k N}(s_0)$ is contained within a single $j$; so, for $k_1$, $k_2$ in distinct remainder classes $\Xi_1, \Xi_2$, we have
\begin{align}
    \fnTransOne_\varphi^{\xi_1 N}(s_0) &= \chi_1 \otimes \rho_{\hat{j}}, \\
    \fnTransOne_\varphi^{\xi_2 N}(s_0) &= \chi_2 \otimes \rho_{\hat{j}}, \\
    \chi_1 \otimes \rho_{\hat{j}} &\perp \chi_2 \otimes \rho_{\hat{j}};
\end{align}

Next, define $U'=U_{\hat{j}} U_{\pi(\hat{j})}\dots U_{\pi^{N-1}(\hat{j})}$; then, for every $k\in\mathbb N_0$,
\[
\fnTransOne_\varphi^{kN}(s_0)=\chi_k\otimes\rho_{\hat\jmath},
\qquad
\chi_k=(U')^k\chi(U'^\dagger)^k .
\]
Since the agent with update $\fnTransOne_\varphi^N$ solves $\mathrm{MCC}(p,\Xi)$,
if $k\bmod p\in\Xi_a$ and $\ell\bmod p\in\Xi_b$ with
$a\neq b$, then these two states must be perfectly
distinguishable. Hence $\chi_k\chi_\ell=0$.

Define
\[
\mathcal K_a
 =\operatorname{span}\bigl\{
   \operatorname{supp}(\chi_k):
   k\geq0,\ k\bmod p\in\Xi_a
  \bigr\}.
\]
Then $\mathcal K_a\perp\mathcal K_b$ whenever $a\neq b$.
Let $\mathcal M'$ be the projective measurement onto the
subspaces $\mathcal K_a$, together with the projector onto
their orthogonal complement. This measurement classifies
every promised input length correctly.

Finally, let $\mathcal{S}'=D(\mathcal{H}_{\hat{j},1})$, let $s'_0=\chi$, and let
\begin{equation}
    \fnTransOne_\varphi'(s)=U' s U'^\dagger
\end{equation}
for all $s \in \mathcal{M}_{d_{\hat{j}}}$. $\fnTransOne_\varphi'$ is unitary, and is essentially the restriction of $\fnTransOne_\varphi^N$ to $\mathcal{M}_{d_{\hat{j}}}\otimes \rho_{\hat{j}}$.

Thus, defining the new agent $\mathcal{A}'=(\mathcal{S}',s'_{init} = s'_0,\mathcal{T}')$, where $\mathcal{T}'=(\fnTransOne_\varphi',\mathcal{M}'),$ it is easy to see that $\mathcal{A}'$ will also solve \taskMCC.

Finally, the embedding $\chi \mapsto\chi \otimes \rho_{\hat{j}}$ identifies the reduced accessible states with a subset of the original accessible states, and so $\mathcal{A}'$ must have memory cost no more than $\mathcal{A}$, completing the proof.
\end{proof}

\section{Optimality of Quantum Agent Construction}
\label{supp:sec:optimal}

In this section, we wish to prove that our construction is optimal over all agents solving $\taskMCC$, i.e. that our construction has lowest memory cost. In the previous section, we proved that it suffices to consider agents with \textit{unitary} (i.e. Kraus rank 1) update operations. We now show that among all agents with unitary update operations, our construction is optimal.

Now, since the $\fnTransOne$ we have constructed in the previous section is unitary, there must exist a unitary matrix $U$ such that $\fnTransOne(\rho)=U\rho U^\dagger.$ We now prove $U$ is $p$-periodic, up to some total phase:
\begin{lemma}
\label{lemma:periodicity}
Consider an agent $\mathcal{A}=(\mathcal{S},s_{init} = s_0,\mathcal{T})$ solving \taskMCC, where $\mathcal{T}=(\fnTransOne,\mathcal{M})$ and $\fnTransOne = U\rho U^{\dagger}$ is unitary. If $\mathcal{A}$ is optimal over all unitary agents, then (on the accesible Hilbert space) $U^p
=
e^{i\zeta}I$ for some $\zeta$ real.
\end{lemma}
\begin{proof}
First, by Lemma \ref{lemma:purity}, we may assume $s_0$ is pure. We will abuse notation slightly and write (one possible) ket state as $\ket{s_0}.$

Now, recall from Lemma \ref{lemma:projective} that $\mathcal{M}$ is projective. That is,
\begin{equation}
    \mathcal{M}(\rho)=\sum_{k=1}^K \Pi_k \rho \Pi_k
\end{equation}
for some projections $\Pi_k$, i.e. $\Pi_k^2=\Pi_k.$ Specifically, whenever the measurement result is $k$, the agent reports the input string is in class $k.$ To be capable of solving $\taskMCC$, when the agent is given an input string of $n$ ones followed by $\Omega$, and $n \equiv \xi \, (\textrm{mod}\, p) \in \Xi_k,$ the measurement $\mathcal{M}$ \textit{must} deterministically return $k$; that is,
\begin{equation}
    \tr \left(\Pi_k U^{n}\proj{s_0}{s_0}{U^{n}}^{\dagger}\Pi_k \right)=1
\end{equation}
Using the cyclic property of the trace this implies that
\begin{equation}\label{eq:unityinnerproduct}
    \ematrix{s_0}{{U^{n}}^{ \dagger}\Pi_k U^{n}}{s_0}=1.
\end{equation}
Thus the inner product of $U^{n}\ket{s_0}$ and $\Pi_k U^n \ket{s_0}$ is 1. Furthermore since $|U^{n}\ket{s_0}|=|\Pi_k U^n \ket{s_0}| = 1$, by the Cauchy-Schwarz inequality applied to Eq. \eqref{eq:unityinnerproduct} we must therefore have
\begin{equation}
\label{eqn:ProjUL}
    \Pi_k U^{n} \ket{s_0} = U^{n} \ket{s_0}.
\end{equation}
Now, clearly $n + p \equiv n \equiv \xi \, (\textrm{mod}\, p)$, so if $\mathcal{A}$ solves $\taskMCC$, we must also have
\begin{equation}\label{eq:sufficency}
    \tr \left(\Pi_k U^{n+p}\proj{s_0}{s_0}{U^{n+p}}^{ \dagger}\Pi_k \right)=1
\end{equation}
and thus
\begin{equation}
\label{eqn:ProjULp}
    \Pi_k U^{n+p} \ket{s_0} = U^{n+p} \ket{s_0}.
\end{equation}
In fact, since $U^n$ and $\Pi_k U^n$ are linear operators, we find that for \textit{any} linear superposition described by $\alpha, \beta,$
\begin{equation}
\label{eqn:ProjULab}
    \Pi_k U^n (\alpha \ket{s_0} + \beta U^p\ket{s_0})=
    U^n (\alpha \ket{s_0} + \beta U^p\ket{s_0}).
\end{equation}

Now, suppose $\agent$ is optimal, with memory dimension $M$, but $U^p\neq e^{i\zeta} I$ for any real $\zeta$. In particular, this implies that there are two eigenvalues $\lambda_1, \lambda_2$ of $U$ (with respective eigenvectors $\ket{\phi_1},\ket{\phi_2} \in \mathcal{H}(\agent)$, as the accessible space is $U$-invariant and equal to the memory space) such that $\lambda_1^p \neq \lambda_2^p.$ $\mathcal{H}(\agent)$ is spanned by the cyclic orbit of $\ket{s_0}$, so every eigenvector in it has nonzero amplitude. Thus, $\ket{s_0}$ and $U^p\ket{s_0}$ are linearly independent. Choose a normalized nonzero linear combination
\begin{equation}
    \ket{s'_0} = \alpha \ket{s_0}+\beta U^p\ket{s_0}
\end{equation}
orthogonal to $\ket{\phi_1}$, and construct an agent $\agent'=(\mathcal{S}',\proj{s'_0}{s'_0},\mathcal{T})$;
\begin{equation}
\mathcal{H}(\agent')
=\mathrm{span}\{U^n\ket{s'_0}:n\geq0\}.\end{equation}
Since $|\phi_1\rangle$ is an eigenvector of $U$,
\begin{equation}
\mathcal{H}(\agent')
\subseteq
\mathcal{H}(\agent)\cap\ket{\phi_1}^\perp,
\end{equation}
so $\dimHacc(\agent')\leq M-1$.
Then equation \eqref{eqn:ProjULab} implies that $\agent'$ will output the same outputs as $\agent$ for any $k$, and will therefore solve $\taskMCC$ with lower memory cost. This is a contradiction: thus, an optimal unitary agent must have $U^p= e^{i\zeta} I$ for some real $\zeta$,  on $\mathcal{H}(\agent)$.

\end{proof}

\begin{lemma}
\label{lemma:diag}
Consider an agent $\mathcal{A}=(\mathcal{S},s_0,\mathcal{T})$ solving \taskMCC, where $\mathcal{T}=(\fnTransOne,\mathcal{M}),$ and $\fnTransOne = U\rho U^{\dagger}$. If $\mathcal{A}$ is optimal over unitary agents,
$U$ must be unitarily equivalent to $e^{i\zeta/p}\mathrm{diag}(\omega^0, \omega^1, \ldots, \omega^{(p-1)})$ or a restriction to a subspace thereof, where $\omega=e^{\ii 2\pi/p}$ is the $p$-th primitive root of unity. i.e. $U$ has no degenerate eigenvalues on the accessible space.
\end{lemma}
\begin{proof}
By the lemma \ref{lemma:periodicity}, $U^p\big|_{\mathcal H_{\rm acc}(\agent)}=e^{\ii\zeta}I.$
Thus, all the eigenvalues of $U\big|_{\mathcal H_{\rm acc}(\agent)}$ must be powers of $\omega$ multiplied by $e^{\ii\zeta/p}$.
It remains to be shown that the multiplicity of each eigenvalue is 1; that is, that the eigenvalues $\lambda_i$ are \textit{distinct}. Suppose we were to pick an orthonormal basis of eigenvectors $\ket{\phi_1}, \ket{\phi_2}, \ldots, \ket{\phi_M}$ with eigenvalues $\lambda_1, \lambda_2, \ldots, \lambda_M$. Then we can always expand $\ket{s_0}$ in this basis,
\begin{equation}
\ket{s_0}=\sum_{i=1}^M \alpha_i \ket{\phi_i}.
\end{equation}
In fact, since this is an eigenbasis,
\begin{equation}
\ket{s_j}=U^j\ket{s_0}=\sum_{i=1}^M \alpha_i \lambda_i^j \ket{\phi_i}.
\end{equation}
Note that this implies that if there is no support over $\ket{\phi_m} $ in the initial state $\ket{s_0}$ then there will be no support over $\ket{\phi_m}$ in any $\ket{s_j} = U^j \ket{s_0}$ achieved by the agent.

Now, suppose without loss of generality that $\lambda_1=\lambda_2$, but $\alpha_1\neq 0, \alpha_2 \neq 0$. We may then rotate $\ket{\phi_1}, \ket{\phi_2}$ into:
\begin{align}
    \ket{\psi_1} &= \frac{\alpha_1\ket{\phi_1}+\alpha_2\ket{\phi_2}}{\sqrt{|\alpha_1|^2+|\alpha_2|^2}} \nonumber\\
    \ket{\psi_2} &= \frac{\alpha_2^*\ket{\phi_1}-\alpha_1^*\ket{\phi_2}}{\sqrt{|\alpha_1|^2+|\alpha_2|^2}}
\end{align}
The initial state evidently has zero overlap on $\ket{\psi_2};$ and because this is an eigenvector decomposition, we can simply remove it and generate a new set of $p$ quantum states in a Hilbert space of dimension $M-1$, contradicting our earlier assumption that our states had minimal dimension. Thus the eigenvalues all have multiplicity 1, proving the lemma.
\end{proof}

\begin{theorem}
\label{thm:optimalquantum}
The quantum agent constructed in Algorithm \ref{alg:construction} is optimal over \textit{all} agents: in particular, no agent solving $\taskMSC$ can have memory dimension less than $M=2|\Delta|+1.$
\end{theorem}
\begin{proof}
First observe that if an agent can solve $\taskMCC$ then it can be used to solve $\taskMSC$ with the same memory cost (and vice versa) by Lemmas \ref{lemma:tally1}-\ref{lemma:tally2}. So we prove these results in the context of solving $\taskMCC$.

Suppose we had some other optimal agent $\mathcal{A}=(\mathcal{S},s_0,\mathcal{T})$ solving $\taskMCC$, where $\mathcal{T}=(\fnTransOne,\mathcal{M})$.
First, by Theorem \ref{thm:unitary}, we may assume $\fnTransOne$ is unitary; so, $\fnTransOne(s)=UsU^\dagger$ for some  unitary $U$.
Using Lemma \ref{lemma:diag}, we can diagonalize $U$, and write $\ket{s_0}$ in its eigenbasis:
\begin{equation}
\ket{s_0}=\sum_{j=0}^{p-1}a_j\ket{\phi_j}
\end{equation}
where $\ket{\phi_j}$ is the eigenvector with eigenvalue $e^{i\zeta/p}\omega^j$ where $\omega = e^{i 2\pi/p}$, and $a_j=0$ if that $e^{i\zeta/p}\omega^j$ is not an eigenvalue of $U$. Let $\ket{s_n}=U^n \ket{s_0};$ by lemma \ref{lemma:delta}, $\braket{s_0}{s_\delta}=\braket{s_0}{s_{p-\delta}}=0$ for all $\delta\in \Delta.$ However, since all the states are generated by the initial state, it is easy to show
\begin{align}
\braket{s_0}{s_\delta}&=\sum_{j=0}^{p-1}|a_j|^2 e^{i\zeta\delta/p}\omega^{j\delta}\\
\braket{s_0}{s_{p-\delta}}&=\sum_{j=0}^{p-1}|a_j|^2 e^{i\zeta(p-\delta)/p}\omega^{-j\delta}
\end{align}
Letting $b_j=|a_j|^2,$ we have a collection of $2|\Delta|$ \textit{complex} equations which $b_j$ must satisfy in order for $\ket{s_0}$ and $\ket{s_\delta}$ to be orthogonal, for all $\delta \in \Delta$. Writing these equations out explicitly, and canceling common factors, we have
\begin{align}
0&=\sum_{j=0}^{p-1}b_j \omega^{j\delta},\label{eq:consdelta}\\
0&=\sum_{j=0}^{p-1}b_j \omega^{-j\delta}.\label{eq:conspmdelta}
\end{align}

Let us concatenate these constraints \eqref{eq:consdelta}--\eqref{eq:conspmdelta} into a matrix equation, $C\mathbf{b}=\mathbf{0}.$ The resulting $2|\Delta|$--by--$p$ matrix of coefficients $C$ is special: by construction, it is equal to a submatrix of the DFT matrix.

Now, suppose only $m\le 2|\Delta|$ of the $b_j$ are nonzero. Then we can identify $(p-m)$ columns of $C$ labeled by indices $\{j \,|\, b_j = 0\}$; let us call these {\it zero-weight columns}. We then construct a new matrix $C'$ by removing $p-2|\Delta|$ of these {\it zero-weight columns} from the matrix $C$. The new matrix $C'$ describes a system of $2|\Delta|$ orthogonality constraints (arising from Equations \eqref{eq:consdelta} and \eqref{eq:conspmdelta}),  being applied to a subspace of $m \le 2|\Delta|$ non-zero variables $b_j$.

By Chebotarev's theorem on roots of unity, since this new matrix of equations is a submatrix of a DFT matrix of  prime order $p$, it must be invertible. Thus, the solution is $b_j = 0$ for \textit{all} $j$. (For an elementary proof of Chebotarev's theorem see \cite{tao2003}.)

However there is one more constraint on our coefficients. In order for our state to be normalized, we require
\begin{equation}
1=\sum_{j=0}^{p-1}b_j\label{eq:btra};
\end{equation}
This implies not all the $b_j$ can be zero, such that we have a contradiction.
 Therefore at least $2|\Delta| + 1$ of the $b_j$ are nonzero.

By Lemma \ref{lemma:diag}, the vectors $\ket{\phi_j}$ appearing with
$a_j\neq0$ are distinct orthonormal eigenvectors in
$\mathcal{H}(\agent)$. Since $b_j=|a_j|^2$, the preceding
bound therefore implies
\begin{equation}
\dim\mathcal{H}(\agent)
\geq \#\{j:b_j\neq0\}
\geq 2|\Delta|+1.
\end{equation}
completing the proof.

\end{proof}

\section{A natural distribution that small classical agents misclassify}
\label{sec:naturalguess}
In theorem \ref{thm:optimalclassical}, we demonstrated that an input distribution exists such that classical agents with memory less than $p$ must guess. We now construct a \textit{natural} distribution with this property. The idea is simple: select from a finite set of different words $\strInput^{(k)}$, then pad the input with a large, random multiple of $p$ 1's.
\begin{theorem}
\label{thm:naturalguess}
Consider the geometric probability distribution $X_{\alpha}$ parameterized by the `stopping probability' $0<\alpha<1$:
\begin{equation}
    X_{\alpha}(n) = \alpha(1-\alpha)^{n}
\end{equation}
for non-negative integer $n \in \mathbb{N}$.  We call $X_{\alpha}(n)$ the `prefix distribution.'

In order to construct our natural distribution $\setData$, define a set of words $\Upsilon$ by selecting a single representative $\strInput^{(k)}=1^{l_k}$ where $l_k>0$ from each of $\intOutputs$ classes, i.e. $\fnClass{\strInput^{(k)}}=k$ for $k=1,2,\ldots,\intOutputs$.
(Of course, every remainder can be represented this way.)
Choose a representative $\strInput'$ from $\Upsilon$ at random. Draw $n$ from the prefix distribution $X_{\alpha}(n)$.  Then create the new word $\strInput$ by concatenating $np$ 1's followed by $\strInput'$.

More precisely,
let $1^n$ denote the word formed by repeating $1$ $n$ times, and let $1^n \strInput$ represent $n$ ones followed by $\strInput$. Then, we define $\setData_\alpha$ such that, when drawing $\strInput \sim \setData_\alpha$,
\begin{equation}
    Pr(\strInput)= \begin{cases}
        \frac{\alpha(1-\alpha)^{n}}{\intOutputs} & \strInput=1^{np}\strInput', n \in \mathbb{N}, \strInput'\in \Upsilon\\
        0 & \mathrm{otherwise}
    \end{cases}
\end{equation}
(Note that this is well-defined, since $\Upsilon$ only contains one word from each class, and distinct classes cannot have words with the same remainder mod $p$.)

Consider some classical agent $\mathcal{A}$ with memory less than $p$. The expected loss per word of $\mathcal{A}$ must be at least $1-1/\intOutputs - \epsilon_{\agent}(\alpha)$, where $\epsilon_{\agent}(\alpha) \rightarrow 0$ as $\alpha \rightarrow 0.$
\end{theorem}

\begin{proof}
Let us consider the expected distribution of internal states of $\mathcal{A},$ conditioned on a particular representative $\strInput' = \strInput^{(k)} \sim \Upsilon.$ That is, we condition on $\strInput=1^{np}\strInput^{(k)}$.  This is equivalent to taking the expectation value over prefix lengths. Consider an input word produced from the natural distribution: its sum will be $np+ \kappa$, where $\kappa=l_k$ is the sum of $\strInput^{(k)}$. Therefore, conditioned on selecting $\strInput^{(k)}$, the expected internal states of $\mathcal{A}$ are
\begin{align}
    \eta_k &\equiv \mathbb{E}_{n\sim X_\alpha}[\fnTransOne^{np+\kappa}(s_0)]\nonumber\\
    &=\alpha\sum_{n=0}^\infty (1-\alpha)^n
    \fnTransOne^{np+\kappa}(s_0).
\end{align}

Now, $\fnTransOne$ is stochastic, so $\fnTransOne^p$ is also; thus, $\fnTransOne^p$ is power-bounded, and thus its peripheral eigenvalues are semisimple;
or, in other words, the peripheral `part' of $\fnTransOne^p$ is diagonalizable.
Let us define
\begin{align}
    R_\alpha = \alpha \sum_{n=0}^\infty (1-\alpha)^n (\fnTransOne^p)^n;
\end{align}
then, the limit $\lim_{\alpha \rightarrow 0} R_\alpha$ exists.

Now, let us consider the limit of $R_\alpha$ as $\alpha \rightarrow 0$.
Furthermore, consider an eigenvector $\psi$ of $\fnTransOne^p$ with eigenvalue $\lambda$.
Consider the action of $R_\alpha$ on this eigenvector as $\alpha \rightarrow 0$.
Let $\beta = 1 - \alpha$; then, $R_\alpha \psi = \frac{1-\beta}{1-\beta \lambda}\psi$.
Thus, as $\alpha \rightarrow 0$, $R_\alpha \psi \rightarrow \psi$ if $\lambda = 1$ and 0 otherwise.
Other generalized eigenspaces of $\fnTransOne^p$ vanish likewise.
In other words, $R_\alpha \rightarrow \Pi$, the spectral projection onto $\rm{Fix}(\fnTransOne^p)$.

Next, let the periods of the communicating class(es) of $\fnTransOne$ be $h_i$.
By the Perron-Frobenius structure of Lemma \ref{lemma:smallclassical}, the peripheral spectrum of (possibly reducible) $\fnTransOne$
is the union over recurrent classes of the $h_i$-th roots of unity, since distinct recurrent classes do not communicate.
Now, we may assume that the entire space on which $\fnTransOne$ acts is accessible.
But each communicating class must sit in the accessible space, and so each eigenvalue $\lambda$ is
an $h_i$-th root of unity for some $h_i \le \dimHacc < p$.

Now, consider a state inside $\mathrm{Fix}(\fnTransOne^p)$. Consider its decomposition into
eigenvectors of $\fnTransOne$. By the previous paragraph, the corresponding eigenvalues must satisfy $\lambda^{h_i}=1$
for some $h_i < p$. However, because the state lies within $\mathrm{Fix}(\fnTransOne^p)$, they
must \textit{also} satisfy $\lambda^p = 1$. Because $p$ is prime, this immediately implies
$\lambda = 1$. Therefore  $\mathrm{Fix}(\fnTransOne^p)=\mathrm{Fix}(\fnTransOne)$, and in particular
$\fnTransOne\Pi = \Pi$.

Now, let $\eta_k = \fnTransOne^{l_k} R_\alpha s_0$. Since $\fnTransOne\Pi = \Pi$,
$\fnTransOne^{l_k}\Pi$ = $\Pi$, and so $\eta_k \rightarrow \Pi s_0$ for \textit{every} $k$.
In other words, in the limit $\alpha \rightarrow 0$, the `final state' is independent of $k$.

The remainder follows as in Theorem \ref{thm:optimalclassical}: by the data processing inequality and Lemma \ref{lemma:guess}, the success probability is at most $1/K + o(1)$ as $\alpha \rightarrow 0$, and so
\begin{equation}
    \fnELoss{\agent} \ge 1 - \frac{1}{K} - \epsilon_{\agent}(\alpha)
\end{equation}
\end{proof}

Thus, we have demonstrated that for a natural family of input distributions, classical agents \textit{must} use at least $p$ memory states to perform better than guessing on the modular counting classification problem, and thus the Black box prize wheel problem.

\vspace{10pt}

\section{Worked example}

Suppose we wished to distinguish sums 0, 1, and 12 from each other mod 13. In this case, $\Xi_1=\{0\}, \Xi_2=\{1\}, $ and $\Xi_{3}=\{12\}$. The set of differences is $\Delta=\lbrace 1, 2\rbrace$. The matrix equation $A\mathbf{b}=\mathbf{e}_0$ can be treated as the equality constraints on a linear program, together with the positivity constraints on $\mathbf{b}$. In order to find a solution with as few nonzero entries as possible, we may proceed using e.g. the simplex method \cite{stone1991simplex} for finding extremal values. After determining which eigenvalues are needed, we can solve the resulting system exactly. One possibility is as follows:

\begin{widetext}
\begin{subequations}
\label{eq:bnexample}
\begin{align}
b_0&=\frac{\cos\bfrac{3\pi}{13} - \sin\bfrac{3\pi}{26}}{3\cos\bfrac{3\pi}{13} + 2\sin\bfrac{\pi}{26} +
    \sin\bfrac{3\pi}{26} + 2\sin\bfrac{5\pi}{26}},\\
b_2&=b_{11}=\frac{\cos\bfrac{3\pi}{13} + \sin\bfrac{\pi}{26}}{3\cos\bfrac{3\pi}{13} + 2\sin\bfrac{\pi}{26} +
    \sin\bfrac{3\pi}{26} + 2\sin\bfrac{5\pi}{26}},\\
b_5&=b_8=\frac{4\cos^2\bfrac{\pi}{26}\cos\bfrac{\pi}{13}
      \sin\bfrac{\pi}{26}}{(1 - \sin\bfrac{\pi}{26})(2 - 2\cos\bfrac{\pi}{13} -
        \cos\bfrac{2\pi}{13} + 5\cos\bfrac{3\pi}{13} + 4\sin\bfrac{\pi}{26} -
              2\sin\bfrac{5\pi}{26})},\\
b_n&=0 \textrm{ otherwise.}
\end{align}
\end{subequations}
\end{widetext}

Thus, we may set
\begin{equation}
\ket{s_0}=\sqrt{b_0}\ket{0} + \sqrt{b_2}\ket{1} + \sqrt{b_5}\ket{2} + \sqrt{b_8}\ket{3} +\sqrt{b_{11}}\ket{4} \notag
\end{equation}
And $U=\textrm{diag}(1,\omega^2,\omega^5,\omega^8,\omega^{11})$ where $\omega=e^{2\pi\ii/13}.$
(Note that while the eigenvalues appear in conjugate pairs in this solution, there are other optimal solutions where these pairs are `broken'.)

Finally, since in this case each $\Xi_j$ corresponds to a single state, for the final measurement $\mathcal{M}$ we can simply construct our measurement basis with $\ket{s_0}, \ket{s_1}, \ket{s_{12}},$ where $\ket{s_n} = U^n \ket{s_0}$. We can then add in two more orthonormal states to make a complete basis.
Since we have already shown these measurement basis states are orthogonal to each other, this completes our description of the quantum agent: a 5-dimensional system distinguishing three classes of words mod 13. This is a significant savings over the classical case, which needs all 13 dimensions.

In fact, this generalizes: distinguishing between $\Xi_1 = \{p-1\}$, $\Xi_2 = \{0\}$, and $\Xi_3 = \{1\}$ $\textrm{mod} \, p$ requires a five-dimensional quantum memory for every odd prime
$p\geq5$, and is therefore an example of unbounded quantum advantage.

\end{appendix}

\end{document}